\documentclass[11pt,reqno]{amsart}
\usepackage[a4paper,width=185mm,top=20mm,bottom=20mm,bindingoffset=6mm]{geometry}
\usepackage{etex}
\usepackage[all]{xy} 
\usepackage{etoolbox} 
\usepackage{lmodern}
\usepackage{array}
\usepackage{xstring}
\usepackage{longtable}
\usepackage[listings]{tcolorbox} 
\usepackage{natbib}
\usepackage{hyperref}
\tcbuselibrary{breakable}
\tcbuselibrary{skins}
\usepackage{amsmath}
\usetikzlibrary{arrows.meta, positioning, shapes.geometric}
\usepackage{fancyhdr}
\usepackage{amsmath}
\usepackage{amssymb}
\usepackage{amsfonts}
\usepackage{pifont}
\usepackage{natbib}
\usepackage{epsfig}
\usepackage{latexsym}
\usepackage{amssymb}
\usepackage{mathtools}
\usepackage{color}
\usepackage{amscd}
\usepackage{multirow}
\usepackage{graphicx}
\usepackage{lscape}
 \usepackage{tabularx}
\usepackage{float}
\usepackage{xcolor}
\usepackage{bm}
\usepackage{tikz-cd}
 \makeatletter
    \let\stdchapter\section
    \renewcommand*\section{%
    \@ifstar{\starchapter}{\@dblarg\nostarchapter}}
    \newcommand*\starchapter[1]{%
        \stdchapter*{#1}
        \thispagestyle{fancy}
        \markboth{\MakeUppercase{#1}}{}
    }
    \def\nostarchapter[#1]#2{%
        \stdchapter[{#1}]{#2}
        \thispagestyle{fancy}
    }
\makeatother

 \makeatletter
\@addtoreset{equation}{section} 
\makeatother

\newtcolorbox{boxA}{
    fontupper = \bf,
    boxrule = 1.5pt,
    colframe = black 
}

\newtheorem*{theorem A}{Main Theorem A}
\newtheorem*{theorem B}{Main Theorem B}
\newtheorem*{theorem C}{Main Theorem C}
\newtheorem*{theorem D}{Main Theorem D}
\newtheorem{theorem}{Theorem}[section]
\newtheorem*{theorem*}{Main Theorem}
\newtheorem{lemma}[theorem]{Lemma}
\newtheorem{proposition}[theorem]{Proposition}
\theoremstyle{definition}
\newtheorem{definition}[theorem]{Definition}
\theoremstyle{plain}
\newtheorem{corollary}[theorem]{Corollary}
\newtheorem{remark}[theorem]{Remark}

\theoremstyle{conclusion}

\usepackage{amsmath}
\usepackage{amssymb}
\usepackage{amsfonts}
\usepackage{esvect}
\usepackage{physics}
\usepackage{mathrsfs}
\usepackage{authblk}
\usepackage{geometry}
\usepackage{abstract}
\usepackage{xcolor}

\begin{document}

\thispagestyle{empty}
\begin{center}
	\Large{{\bf  Geometric construction of superintegrable Poisson projection chains via Poisson centralizers }}
\end{center}
 
\vskip 0.3cm
\begin{center}
	\textsc{Kai Jiang$^{1\star}$,  Guorui Ma$^{2a,2b,\sharp}$,    Ian Marquette$^{3,\bullet}$,   Junze Zhang$^{4,\dagger}$  and Yao-Zhong Zhang$^{4,\ddagger}$}
\end{center}
\vskip 0.2cm
\begin{center}
	$^1$   Paris Curie Engineer School, Beijing University of Chemical Technology, Beijing, China
\end{center}
\begin{center}
	$2^a$ Shanghai Institute for Mathematics and Interdisciplinary Sciences (SIMIS), Shanghai, 200433, China  \\
    $2^b$  Research institute of Intelligent Complex Systems, Fudan University, Shanghai, 200433, China 
\end{center}
\begin{center}
	$^3$ Department of Mathematical and Physical Sciences, La Trobe University, Bendigo, VIC 3552, Australia
\end{center}

\begin{center}
	$^4$ School of Mathematics and Physics, The University of Queensland, Brisbane, QLD 4072, Australia
\end{center}

\begin{center}
	\footnotesize{ $\star$\textsf{kai.jiang.math@gmail.com} \hskip 0.25cm
    $\sharp$\textsf{mgr18@tsinghua.org.cn} \hskip 0.25cm
    $\bullet$\textsf{i.marquette@latrobe.edu.au} \hskip 0.25cm
 $^\dagger$\textsf{junze.zhang@uq.net.au} \hskip 0.25cm
 $^\ddagger$\textsf{yzz@maths.uq.edu.au}}
\end{center}
\vskip  1cm

\begin{abstract}
\noindent We introduce a geometric framework for constructing superintegrable systems from Poisson centralizers (commutants) in the Lie-Poisson algebra $S(\mathfrak{g})$ of a complex semisimple Lie algebra. Starting from a chain of reductive subgroups, we study the corresponding invariant Poisson subalgebras and their Poisson centers, and formulate superintegrability in terms of a \emph{Poisson projection chain} of affine Poisson varieties. For a maximal torus $T\subset G$, we prove that the inclusions $S(\mathfrak{g})^G\subset S(\mathfrak{g})^T\subset S(\mathfrak{g})$ determine a superintegrable chain and identify the associated quotient maps $\mathfrak{g}\xrightarrow{\chi_T}\mathfrak{g}//T\xrightarrow{\rho}\mathfrak{g}//G$. The rank (transcendence degree) computations yield the expected dimension split between commuting Hamiltonians and first integrals, and we describe the corresponding symplectic leaves in the intermediate space. Several examples illustrate how the centralizer generators organize into explicit superintegrable Poisson chains.
\end{abstract}
\vskip 0.35cm
\hrule

\tableofcontents

\section{Introduction}
\label{sec:introduction} 

Integrable and superintegrable Hamiltonian systems provide a unifying language for the appearance of large families of conserved quantities in mathematical physics and symplectic geometry \cite{MR3119484,MR3493688,MR187763,nehorovsev1968action,MR3942135,MR4071113,MR4282982,MR4644061}. For a $2n$-dimensional symplectic manifold, Liouville integrability requires $n$ functionally independent first integrals that Poisson-commute, leading (on regular level sets) to $n$-dimensional invariant Lagrangian tori. Superintegrability strengthens this by allowing more than $n$ independent integrals, so that the generic invariant tori have smaller dimension $k<n$. In the extreme case of maximal superintegrability, one has $2n-1$ integrals and all trajectories are closed \cite{MR3493688,MR3942135,MR2023556}. The perspective adopted in this paper is that such over complete families of integrals can be organized algebraically via Poisson centralizers in the symmetry algebra $S(\mathfrak{g})$ and geometrically via Poisson projection chains between the associated affine Poisson quotients.

The modern notion of superintegrability, as formalized by Nekhoroshev \cite{nehorovsev1968action}, extends the Liouville-Arnold theorem to the noncommutative case and provides geometric descriptions of how superintegrabilities are realized in symplectic manifolds.  A Liouville integrable system consists of $n$ functionally independent conserved quantities on a $2n$-dimensional symplectic manifold and yields action-angle coordinates \cite{nehorovsev1968action,MR997295}. The further development on constructing superintegrable systems can be understood from the following two perspective:

On the algebraic side, a series of examples related to Lie groups was given by Mishchenko and Fomenko in \cite{Mishchenkofomenko1978}. In their paper, they introduced the argument shift method in the study of Euler equations on Lie algebras, constructing large Poisson-commutative algebras of integrals from invariant polynomials. Independently, Thimm developed a complementary method based on chains of subalgebras and proved that invariant geodesic flows in certain homogeneous spaces are completely integrable \cite{Thimm1981}.

On the geometric side, Guillemin and Sternberg\cite{GuilleminSternberg1983} constructed the Gelfand-Cetlin system on coadjoint orbits of $U(n)$ (in particular, on flag manifolds): using the chain $U(1)\subset\cdots\subset U(n)$ and the corresponding moment maps, one obtains the Gelfand-Cetlin map whose components form an explicit collection of commuting Hamiltonians and yield a completely integrable system on a dense open subset of the orbit. 

These first integrals in these constructions are expressed as polynomials in a basis of the Lie algebra. It is not required that they all Poisson-commute, provided that they generate a commutative subalgebra of sufficiently large dimension. This generalized integrability framework has been employed in a number of contexts. For instance, noncommutative integrability of certain magnetic geodesic flows on homogeneous symplectic manifolds was established by Efimov \cite{MR2141306}. A construction based on Poisson projection chains was formalized by Reshetikhin \cite{MR3493688}, and equivalence relations between superintegrable systems were derived. More recently, new superintegrable systems have been constructed on moduli spaces and in other Lie-theoretic settings \cite{MR4299124}.  

Although the theoretical framework of integrable and superintegrable systems is relatively mature, significant research directions remain in the following directions:

(1) Geometric realization of superintegrability: While Nekhoroshev's theory provides a general framework for noncommutative integrability, the explicit construction of families of conserved quantities satisfying the noncommutative integrability conditions on concrete homogeneous or symmetric spaces remains an open problem.

(2) Classification and hierarchical structure of superintegrable systems: There exist various intermediate cases between minimally superintegrable ($n+1$ integrals) and maximally superintegrable ($2n-1$ integrals). How can these different levels of superintegrability be characterized geometrically? What is their correspondence with the foliation structure of the phase space?

(3) Unification of algebraic and geometric methods: The argument shift method of Mishchenko-Fomenko and the subalgebra chain method of Thimm construct integrable systems from algebraic and geometric perspectives, respectively. However, the intrinsic relationship between these two methods, as well as their generalization to superintegrable cases, has not been fully explained.

(4) Construction of superintegrable systems in new contexts: Recent progress has been made in constructing novel superintegrable systems on moduli spaces and in other Lie-theoretic contexts \cite{MR4299124}. Nevertheless, the explicit expressions of their first integrals, their relationship with classical superintegrable systems, and their behaviour under quantization remain to be thoroughly investigated.

These questions suggest that, beyond individual examples, one needs a structural mechanism that simultaneously produces large families of first integrals, controls the induced foliation of the phase space, and clarifies how algebraic constructions in $S(\mathfrak{g})$ manifest geometrically on Poisson quotients. A natural setting where these topic meet is provided by 
reductive chains of Lie algebras and the associated chains of Poisson subalgebras. 
 Recently, the construction of superintegrable systems based on Lie algebra reductive chains has been studied in \cite{campoamor2026construction,marquette2023algebraic,campoamor2023algebraic}. The main novelties of the present paper are:
\begin{itemize}
\item We make explicit the geometric construction from Poisson subalgebra chains in $S(\mathfrak{g})$ to canonical sequences of Poisson maps between affine Poisson quotients, including the dimension identities required for superintegrability.
\item From suitable Poisson subalgebras of $S(\mathfrak{g})$, we construct an explicit \emph{Poisson projection chain} (Definition \ref{def:superge}), and prove that it satisfies the superintegrability dimension conditions, recovering known special cases and producing new families of superintegrable systems.
\item We explain how algebraic invariant theory controls the geometry of the induced foliations and symplectic leaves along the projection chain.
\item We provide a unified Lie-theoretic framework that connects reductive (subalgebra) chain constructions with Mishchenko-Fomenko noncommutative integrability.
\end{itemize}

To help the reader navigate the paper, we summarize the main constructions and theorems below.  The precise results are as follows.

(I) Theorem \ref{thm:superalchain} is the basic example that motivates the rest of the paper: taking $A=T$ produces a canonical Poisson inclusion chain $S(\mathfrak{g})^G\subset S(\mathfrak{g})^T\subset S(\mathfrak{g})$ and hence a Poisson projection chain $\mathfrak{g}\to\mathfrak{g}//T\to\mathfrak{g}//G$. In particular, the base algebra is forced to be $S(\mathfrak{g})^G$, which Poisson-commutes with all $T$-invariants, and the dimension identity is governed by the rank $r=\dim T$.

\begin{theorem A} 
(Theorem \ref{thm:superalchain}) 
Let $G$ be a connected complex semisimple Lie group with Lie algebra $\mathfrak{g}$, and let $T\subset G$ be a maximal torus of rank $r$. With the Lie-Poisson bracket on $S(\mathfrak{g})=\mathbb{C}[\mathfrak{g}^*]$, $S(\mathfrak{g})^G\subset S(\mathfrak{g})^T\subset S(\mathfrak{g})$ is a Poisson subalgebra chain inducing the superintegrable affine Poisson chain
\begin{equation}
\mathfrak{g}\xrightarrow{\chi_T}\mathfrak{g}//T\xrightarrow{\rho}\mathfrak{g}//G
\label{eq:supgtc}
\end{equation}
in the sense of Definition \ref{def:superalge}. Moreover, \begin{align*}
S(\mathfrak{g})^G\subset\mathcal{Z}\bigl(S(\mathfrak{g})^T\bigr),\quad
\mathrm{trdeg}_{\mathbb{C}}S(\mathfrak{g})^G=r,\quad
\mathrm{trdeg}_{\mathbb{C}}S(\mathfrak{g})^T=\dim\mathfrak{g}-r,  
\end{align*} and $\dim\mathrm{Spec}\,S(\mathfrak{g})=(\dim\mathfrak{g}-r)+r$. Here $\chi_T$ and $\rho$ are induced by the corresponding inclusions of invariant algebras.
\end{theorem A}

(II) Our second main result replaces the maximal torus with an arbitrary Abelian reductive and connected subgroup $A \subset G$. In this setting, there is, in general, no canonical analogue of $S(\mathfrak{g})^G$ inside the Poisson center of $S(\mathfrak{g})^A$. Instead, we construct a natural central subalgebra from the linear moment map $\mu_A\colon\mathfrak{g}\to\mathfrak{a}^*$ induced by an $\mathrm{Ad}(G)$-invariant bilinear form and show that it provides the required base algebra for superintegrability.

\begin{theorem B}
  (Theorem \ref{thm:abelianredsuper})
Let $G$ be a connected complex semisimple Lie group with Lie algebra $\mathfrak{g}$, and let $A\subset G$ be connected Abelian reductive with Lie algebra $\mathfrak{a}$, $\dim\mathfrak{a}=s$. Fix a nondegenerate $\mathrm{Ad}(G)$-invariant symmetric bilinear form $B$, define \begin{align*}
    \mu_A(X)(H):=B(X,H),\qquad \mathcal{B}_A:=\mu_A^*\mathbb{C}[\mathfrak{a}^*]\subset S(\mathfrak{g}),
\end{align*} where $\mu_A:\mathfrak{g}\to\mathfrak{a}^*$. Then \begin{align*}
\mathcal{B}_A\subset\mathcal{Z}\bigl(S(\mathfrak{g})^A\bigr),\qquad
\mathrm{trdeg}_{\mathbb{C}}\mathcal{B}_A=s,\qquad
\mathrm{trdeg}_{\mathbb{C}}S(\mathfrak{g})^A=\dim\mathfrak{g}-s.
\end{align*} Hence $\mathcal{B}_A\subset S(\mathfrak{g})^A\subset S(\mathfrak{g})$ is superintegrable in the sense of Definition \ref{def:superalge}, equivalently inducing \begin{align*}
\mathfrak{g}\xrightarrow{\chi_A}\mathfrak{g}//A\xrightarrow{\overline{\mu}_A}\mathfrak{a}^*, \qquad \mu_A=\overline{\mu}_A\circ\chi_A,
\end{align*} where $\chi_A$ is the affine quotient morphism.
\end{theorem B}

(III) The first two theorems are formulated purely in terms of Poisson algebras on $\mathfrak{g}$ and their affine quotients. The next step is to connect this Lie theoretic construction with the symplectic geometry of the cotangent bundle $T^*M$ and the magnetic geodesic flows studied in our previous work \cite{jiang2026poisson,campoamor2026construction}. In those papers, the basic integrals were obtained from the momentum map $P\colon T^*M\to\mathfrak{g}^*$ together with invariant polynomials on $\mathfrak{g}$, and the resulting level sets were analysed via reduction.

Here we package the same families of first integrals into a functorial statement: we construct an intermediate Poisson space $Y : = \mathfrak{g} \times_{\mathfrak{g}//G} (\mathfrak{m}-\varepsilon W)// T$ and a Poisson projection chain $T^*(G/T)\to Y\to\mathfrak{g}//G$ whose spectrum is identified with the torus quotient chain $\mathfrak{g}\to\mathfrak{g}//T\to\mathfrak{g}//G$. This produces a precise spectral equivalence between the dynamical system on $T^*M$ and the algebraic superintegrable system on $\mathfrak{g}$, and it is the framework that allows us to transport information about symplectic leaves and reduced spaces between the two settings.

\begin{theorem C}
 (Theorem \ref{thm:spectralequivalent})
Let $M=G/T$, write $T^*M$ as pairs $(g,u)$ with $u\in\mathfrak{m}-\varepsilon W$, and put $P(g,u):=\mathrm{Ad}(g)u\in\mathfrak{g}^*\cong\mathfrak{g}$. For basic generators $C_1,\ldots,C_r$ of $S(\mathfrak{g})^G$, set $\chi=(C_1,\ldots,C_r)$ and \begin{align*}
Y:=\mathfrak{g}\times_{\mathfrak{g}//G}(\mathfrak{m}-\varepsilon W)//T =\bigl\{(X,[u]_T):\chi(X)=\chi(u)\bigr\}.
\end{align*} With $\pi_1(g,u):=(P(g,u),[u]_T)$, $\pi_2(X,[u]_T):=\chi(X)$, the chain $T^*M\xrightarrow{\pi_1}Y\xrightarrow{\pi_2}\mathfrak{g}//G$ is superintegrable and is spectrally equivalent, in the sense of Definition \ref{def:equidefi}, to \begin{align*}
\mathfrak{g}\xrightarrow{\chi_T}\mathfrak{g}//T\xrightarrow{\rho}\mathfrak{g}//G.
\end{align*} The equivalence is given by $(\phi,\phi_1,\phi_2)=(P,\chi_T\circ\mathrm{pr}_1,\mathrm{id}_{\mathfrak{g}//G})$ and satisfies \begin{align*}
   \chi_T\circ\phi=\phi_1\circ\pi_1,\qquad \rho \circ\phi_1=\phi_2\circ\pi_2=\pi_2.  
\end{align*}
\end{theorem C}

(IV) Since we identified the two projection chains at the level of their spectral data, the next natural question is how the associated Poisson geometry is reflected on each side. In particular, the intermediate quotient $\mathfrak{g}//T$ carries a stratified Poisson structure, and on its regular locus, one expects a concrete description of symplectic leaves in terms of the commuting families of functions that define the superintegrable system. The following theorem makes this description explicit and relates it to the familiar Marsden-Weinstein reduction on coadjoint orbits.

\begin{theorem D}
(Theorem \ref{thm:leaves})
   Let \begin{align*}
    J = (\rho,\overline{\mu}_T):(\mathfrak{g}//T)_{\mathrm{reg}}\to(\mathfrak{g}//G)\times\mathfrak{t}^*, \qquad J([u]_T)=\bigl((C_1(u),\ldots,C_r(u)),u|_{\mathfrak{t}}\bigr).
\end{align*} Fix $c=(c_1,\ldots,c_r)\in(\mathfrak{g}//G)_{\mathrm{reg}}$ and a regular value $\alpha\in\mathfrak{t}^*$ of $\mu_{T,c}:\mathcal{O}_c:=\chi^{-1}(c)\to\mathfrak{t}^*$. Then each smooth connected component of $J^{-1}(c,\alpha)$ is a symplectic leaf of $(\mathfrak{g}//T)_{\mathrm{reg}}$, namely \begin{align*}
    \mathcal{L}_{c,\alpha} =\bigl\{[u]_T\in\mathfrak{g}//T:
C_i(u)=c_i\ (1\leq i\leq r),\ u|_{\mathfrak{t}}=\alpha\bigr\}_{\varpi}.
\end{align*} Equivalently, if $Z_{c,\alpha}:=\mu_{T,c}^{-1}(\alpha)=\mathcal{O}_c\cap\mu_T^{-1}(\alpha)$ is smooth and $T$ acts freely on it, then \begin{align*}
    \mathcal{L}_{c,\alpha}\cong Z_{c,\alpha}/T =\bigl(\mathcal{O}_c\cap\mu_T^{-1}(\alpha)\bigr)/T.
\end{align*} 
\end{theorem D}

 The structure of the paper is as follows: we begin in Section \ref{sec:prel} by reviewing the necessary background on superintegrable systems in symplectic manifolds and setting up the notation. Key definitions and classical results are summarized to provide a foundation for our work. Later in Section \ref{sec:supersg}, we provide the superintegrability of the semisimple Lie algebra $\mathfrak{g}$ with the reductive subgroup $A$ as the Poisson centraliser. Moreover, from the superintegrable chain \eqref{eq:supgtc}, we further deduce the symplectic leaf $\mathcal{L}_{c,\alpha}$ in $\mathfrak{g}//T$ given by the pullback of the regular element in $\mathfrak{g}//G$ via $\rho^{-1}$.  Also, for the reductive and connected subgroup $A$, we provided the case where $A$ is an Abelian subgroup and its corresponding superintegrable chain in Subsection \ref{subsec:AnotequalT}. Then, in Subsection \ref{subsec:mfinsideA}, we discussed the conditions under which the Mishchenko-Fomenko subalgebras can be considered a base algebra. Eventually, in Section \ref{sec:example}, we will applied these geometric constructions to the example $\mathfrak{g} = \mathfrak{sl}(n,\mathbb{C})$, for which the Cartan commutant has been illustrated in \cite{campoamor2023algebraic}. Finally, in Conclusion \ref{sec:conclusion}, we provide a summary of this work and propose some further directions that we would like to extend from this paper.  

\section{Preliminary}
\label{sec:prel}
In this Section \ref{sec:prel}, we focus on the basic background of the Hamiltonian system and provide the definition of the superintegrable Poisson projection chain. For more detailed information, we refer the reader to \cite{MR2906391,MR3493688}. 
Let $(M,\omega)$ be a $2n$-dimensional symplectic manifold. The symplectic form induces the Poisson bracket$\{\cdot,\cdot\}_\omega$ defined by
\begin{equation*}
    \begin{split}
        \{\cdot,\cdot\}_\omega : C^\infty(M)\times C^\infty(M) &\longrightarrow C^\infty(M)\\
        (f,g) &\longmapsto \{f,g\}_\omega := \omega(X_f,X_g),
    \end{split}
\end{equation*}
for $f,g\in C^\infty(M)$, where $X_{f}$ is the Hamiltonian vector field characterized by $\iota_{X_f}\omega = d f$.

A function $g\in C^\infty(M)$ is called a \textit{first integral} (also: integral of motion, conserved quantity) of $f$ if $\{f,g\}_\omega=0$. Equivalently, $g$ is constant along the trajectories of $X_f$ since
\begin{align*}
X_f(g)=\{f,g\}_\omega=0.
\end{align*}
In this sense, first integrals encode symmetries of the Hamiltonian flow generated by $f$.

Let $\mathfrak{g}$ be a finite-dimensional real Lie algebra with Lie bracket $[\cdot,\cdot]$. There is a classical Lie-Poisson bracket on the dual of a Lie algebra $\mathfrak{g}^*$. For $f,g\in C^\infty(\mathfrak{g}^*)$ and $\xi\in \mathfrak{g}^*$, the Lie-Poisson bracket is
\begin{align}
\{f,g\}(\xi)=\langle \xi, [d_\xi f, d_\xi g] \rangle .\label{eq:poiss}
\end{align}
where $d_\xi f, d_\xi g \in  T_\xi^*\mathfrak{g}^*\cong \left(\mathfrak{g}^*\right)^* \cong \mathfrak{g}$, and $\langle \cdot,\cdot \rangle : \mathfrak{g}^*\times \mathfrak{g} \rightarrow \mathbb{C}$ is the dual pair between $\mathfrak{g}^*$ and $\mathfrak{g}$. We refer the reader to  \cite[Chapter 7]{MR2906391} for details. 

For a symplectic manifold $M$, we assume that the Poisson algebra $C^\infty(M)$  has a Poisson subalgebra $\mathcal{A} \subset C^\infty(M)$ of rank (transcendence degree) $2n-k$ for some integer $k$ with $1\leqslant k\leqslant n$ and the Poisson center $\mathcal{Z}(\mathcal{A})$ of $ \mathcal{A}$ is of rank $k$.  Smooth functions in $\mathcal{Z}(\mathcal{A})$ are called Hamiltonian, which are smooth functions Poisson-commuting with first integrals. 



We recall the geometric description of the superintegrability of a Hamiltonian dynamical system. The reader can refer to \cite{MR3493688} for more details and example. 



\begin{definition} 
\label{def:superge} 
Let $(M,\omega)$ be a $2n$-dimensional connected symplectic manifold. A \textit{superintegrable system} on $(M,\omega)$ is a sequence consisting of a triple $(M,\mathcal{P}_{2n-k},\mathcal{B}_k)$ of Poisson manifolds with two Poisson maps $\pi_1$ and $\pi_2$ \begin{align}
M \xrightarrow{\;\pi_1\;} \mathcal{P}_{2n-k} \xrightarrow{\;\pi_2\;} \mathcal{B}_k \label{eq:sup}
\end{align} such that the following holds:

(i) The Poisson maps  $\pi_i$ are Poisson submersions. 


(ii) $\dim \mathcal{P}_{2n-k}+\dim \mathcal{B}_k = 2n = \dim M$. 
\end{definition}
 \begin{remark}
 \label{rkm:ijc}
From the chain in Definition \ref{def:superge}, we obtain two natural Poisson subalgebras of $C^\infty(M)$ by pullback: \begin{align*}
    J=\pi_1^*\bigl(C^\infty(\mathcal{P}_{2n-k})\bigr), \qquad I=(\pi_2\circ\pi_1)^*\bigl(C^\infty(\mathcal{B}_k)\bigr).
\end{align*} Intuitively, elements of $J$ only depend on the $\mathcal{P}_{2n-k}$-coordinate , whereas elements of $I$ only depend on the $\mathcal{B}_k$-coordinate.

For a Poisson subalgebra $A\subset C^\infty(M)$, write \begin{align*}
    C_A(M):=\{f\in C^\infty(M) \ :  \{f,A\}=0\}
\end{align*} for its Poisson commutant. With this notation, we set $C_I(M):=C_I(M)$ and $C_J(M):=C_J(M)$.
Because $\pi_1$ is Poisson and the functions coming from $\mathcal{B}_k$ are Casimirs on $\mathcal{P}_{2n-k}$ (via $\pi_2$), the two pullback algebras Poisson-commute. Equivalently,  $J\subset C_I(M)$  and $I\subset C_J(M)$. Concretely, for $h\in C^\infty(\mathcal{P}_{2n-k})$ and $g\in C^\infty(\mathcal{B}_k)$, \begin{align*}
   \{\pi_1^*h,(\pi_2\circ\pi_1)^*g\}_M =\pi_1^*\bigl(\{h,\pi_2^*g\}_{\mathcal{P}_{2n-k}}\bigr) = 0. 
\end{align*} 
 \end{remark}

 In Section \ref{sec:supersg}, we will introduce an algebraic way to define a superintegrable system based on the Poisson centralizers defined in Remark \ref{rkm:ijc}.  To conclude this section,  for any two superintegrable systems, we have the following way to determine whether they are equivalent or not. 

\begin{definition} \cite{MR3493688}
\label{def:equidefi}
 Let $(M,\omega)$ and $(M',\omega')$ be symplectic manifolds with dimensions $2n$ and $2n'$, respectively, 
 and let \begin{align}
M \xrightarrow{\ \pi_1\ } \mathcal{P}_{2n-k}\xrightarrow{\ \pi_2\ } \mathcal{B}_k \  \text{ and } \  M' \xrightarrow{\pi_1'} \mathcal{P}_{2n'-k'}' \xrightarrow{\pi_2'} \mathcal{B}_{k'}' \label{eq:twosystem}
\end{align} be superintegrable systems on $M$ and $M'$, respectively. Two such systems \eqref{eq:twosystem} are \textit{spectrally equivalent} if there exist Poisson maps \begin{align*}
\phi:M \rightarrow  M',\qquad \phi_1:\mathcal{P}_{2n-k}\rightarrow  \mathcal{P}_{2n'-k'}',\qquad \phi_2:\mathcal{B}_k \rightarrow  \mathcal{B}_{k'}'
\end{align*} with $\phi_2$ a diffeomorphism, making the diagram commute:  \[\begin{tikzcd}
	M & {\mathcal{P}_{2n-k}} & {\mathcal{B}_k} \\
	{M'} & {\mathcal{P}_{2n'-k'}'} & {\mathcal{B}_{k'}'}
	\arrow["\pi_1", from=1-1, to=1-2]
	\arrow["{\phi}", from=1-1, to=2-1]
	\arrow["\pi_2", from=1-2, to=1-3]
	\arrow["{\phi_1}", from=1-2, to=2-2]
	\arrow["{\phi_2}", from=1-3, to=2-3]
	\arrow["{\pi_1'}", from=2-1, to=2-2]
	\arrow["{\pi_2'}", from=2-2, to=2-3]
    \arrow[from=1-1, to=2-2, phantom, "\circlearrowright" {anchor=center, scale=1.5, rotate=90}]
    \arrow[from=1-2, to=2-3, phantom, "\circlearrowright" {anchor=center, scale=1.5, rotate=90}]
\end{tikzcd}\] 
with the following commutative relations:
\begin{align*}
\pi_1'\circ \phi=\phi_1\circ \pi_1,\qquad \pi_2'\circ \phi_1=\phi_2\circ \pi_2.
\end{align*} They are \textit{equivalent} if, in addition, $\phi$ and $\phi_1$ are diffeomorphisms. 
\end{definition}

\begin{remark} 
\label{rem:nonreguequiv}
Note that in Definition \ref{def:equidefi}, we only use the Poisson brackets on $M,\mathcal{P},\mathcal{B}$ and the commutative diagram, while the symplectic forms $\omega,\omega'$ are not implemented. Hence, we are able to extend the notion of spectral equivalence to Poisson manifolds by replacing $(M,\omega)$ and $(M',\omega')$ with Poisson manifolds $(M,\{\cdot,\cdot\})$ and $(M',\{\cdot,\cdot\}')$, and requiring that $\varphi,\varphi_1,\varphi_2$ be Poisson maps, with $\varphi_2$ being a (diffeomorphic) Poisson isomorphism.

\end{remark}

\section{Superintegrable systems on \texorpdfstring{$S(\mathfrak{g})$}{S(\mathfrak{g})}}
 \label{sec:supersg}
 In Section \ref{sec:supersg}, we will construct a superintegrable system on the Poisson manifold $S(\mathfrak{g})$.  In Subsection \ref{subsec:poicentra}, we define the Poisson centralizer with respect to the subalgebra chain $\mathfrak{a} \subset \mathfrak{g}$, which is also known as \textit{the commutant} in the physics literature. See, for instance, \cite{campoamor2023algebraic}. The corresponding algebraic superintegrable system is also defined. Then, in Subsection \ref{subsec:geomi}, we will provide the spectrum of these algebras, which illustrates the geometric interpretation of the algebraic superintegrable system. We see that such constructions lead to Poisson chains as defined in \eqref{def:superge}. In Subsection \ref{subsec:AnotequalT} In Subsection \ref{subsec:mfinsideA} Finally, in Subsection \ref{sec:leaves}, we describe the symplectic leaves in the intermediate Poisson manifold of the corresponding superintegrable chain constructed in Subsection \ref{subsec:geomi}.

\subsection{Poisson centralizer \texorpdfstring{$S(\mathfrak{g})^A$}{S(\mathfrak{g})^A}}
\label{subsec:poicentra}

In this Subsection \ref{subsec:poicentra}, we further assume that $\mathfrak{g}$ is an $n$-dimensional semisimple Lie algebra over $\mathbb{C}$. Suppose that $\beta_\mathfrak{g} = \{X_1,\ldots,X_n\}$ is an ordered basis for $\mathfrak{g}$, which satisfies the relations $$ [X_i,X_j] = \sum_{k=1}^n C_{ij}^k X_k $$ for all $ 1 \leq i,j \leq n $, where $C_{ij}^k \in \mathbb{C} $ are the structure constants of $\mathfrak{g}$.   We now provide the coordinate representation of the Lie-Poisson bracket \eqref{eq:poiss} in $\mathfrak{g}^*$. Define the corresponding coordinate functions $  (x_1,\ldots,x_n)$ on $\mathfrak{g}^*$ by $x_i(\xi) = \langle \xi,X_i\rangle$ with $i = 1,\ldots,n$. For any coordinate functions $x_i,x_j$, using \eqref{eq:poiss}, one has   \begin{equation}
\{x_i,x_j\}= \sum_{k=1}^n C_{ij}^k x_k \qquad 1 \leq i,j \leq n .
\end{equation} 
Choosing the canonical linear embedding of $\mathfrak{g}$ into the algebra of polynomial functions on its dual space $\mathfrak{g}^*$, we obtain a canonical algebra isomorphism  $S(\mathfrak{g}) \cong \mathbb{C}[\mathfrak{g}^*] := \mathbb{C}[x_1,\ldots,x_n]$, where the coordinate functions $ x_j:\mathfrak{g}^*\rightarrow\mathbb{C}$ correspond to the basis elements of $\mathfrak{g}$. In the rest of this paper, we work on the Poisson polynomial ring $\mathbb{C}[\mathfrak{g}^*]$.  Hence, in the following, each $p \in S(\mathfrak{g})$ is a polynomial function from $\mathfrak{g}^*$ to $\mathbb{C}$. For any $p,q \in S(\mathfrak{g}),$ a Lie-Poisson bracket $\{\cdot,\cdot\}:S(\mathfrak{g}) \times S(\mathfrak{g}) \rightarrow S(\mathfrak{g})$ is defined by \begin{equation}
\{p,q\}= \sum_{  i,j,k = 1}^n C_{ij}^k x_k \dfrac{\partial p }{\partial x_i }\dfrac{\partial q}{\partial x_j } .\label{eq:poi}
\end{equation} 

 We now construct the invariant symmetric algebra for an arbitrary subalgebra $\mathfrak{a}\subseteq\mathfrak{g}$. Without loss of generality, we may assume that the basis of $\mathfrak{a}$ is given by $X_{\ell_1},\ldots,X_{\ell_s}$, where $s=\dim\mathfrak{a}$ and $\{\ell_1,\ldots,\ell_s\}\subseteq\{1,\ldots,n\}$. 
 We now examine the adjoint action of the subalgebra $\mathfrak{a}$ on $S(\mathfrak{g})$.
 
 \begin{definition}
 \label{def:p2.1}
 Let $\mathfrak{g}$ be a Lie algebra with its dual $\mathfrak{g}^*$, and let $\mathfrak{a}\subseteq\mathfrak{g}$ be a Lie subalgebra. Then, for any $\xi\in\mathfrak{g}^*$, the infinitesimal coadjoint representation of $\mathfrak{a}$ on $S(\mathfrak{g})$, defined by \begin{align*}
     (X\cdot p)(\xi)=\frac{d}{dt}\bigg\vert_{t=0}p(\mathrm{Ad}^*(\exp (tX))\xi) ,\quad X\in\mathfrak{a},\, p\in S(\mathfrak{g}) 
 \end{align*}  acts by derivations.
 That is, for all $X_m\in\mathfrak{a}$ and $p,q\in S(\mathfrak{g})$,  \begin{align*}
     X_m\cdot\{p,q\}=\{X_m\cdot p,q\}+\{p,X_m\cdot q\}, \quad m = \ell_1,\ldots,\ell_s
 \end{align*} 
 \end{definition}
  \begin{remark}
       The Lie algebra $\mathfrak{a}$ acts infinitesimally on $S(\mathfrak{g})$ by defining the following vector fields: \begin{align}
    \tilde{X}_m=\sum_{k,l=1}^n C_{mk}^l\, x_l\frac{\partial}{\partial x_k},\quad m=\ell_1,\dots,\ell_s.\label{eq:dual}.
\end{align} 
  \end{remark}

According to Definition \ref{def:p2.1}, and its remark, our primary focus is to explore the kernel of the coadjoint action of $\mathfrak{a}$ on $S(\mathfrak{g})$.

\begin{definition}
\label{def:pocen}
 The $\textit{commutant}$ (a.k.a, $\textit{centralizer subalgebra}$) $ S(\mathfrak{g})^\mathfrak{a}$ is defined as the centralizer of $\mathfrak{a}^*$ in $S(\mathfrak{g})$:
\begin{align*}
      S(\mathfrak{g})^\mathfrak{a}   =& \left\{p \in S(\mathfrak{g}): \text{ } \{x,p\} = 0  \quad \text{ for all } x \in \mathfrak{a}^*\right\}.
  \end{align*} 
\end{definition}
\begin{remark}
\label{2.2}
 (i) The $\textit{Poisson center}$ of $\left(S(\mathfrak{g}),\{\cdot,\cdot\}\right)$ is the set of all $\mathfrak{g}$-invariant polynomials, i.e., \begin{align*}
    S(\mathfrak{g})^\mathfrak{g}  =\left\{ p \in S(\mathfrak{g}) : \text{ } \{p,x\} = 0 \quad \text{ for all } x \in \mathfrak{g}^*\right\}.
\end{align*} These elements can be identified with the (polynomial) kernels of the differential operators in \eqref{eq:dual}  (see \cite{MR4355741,MR1451138} for details). These are the classical Casimir polynomials. 

 (ii) Note that for any finite-dimensional Lie algebra,  $S(\mathfrak{g})^\mathfrak{a}$ is not necessarily finitely generated. However, if $\mathfrak{g}$ is semisimple or reductive (note that in this paper, we assume that $\mathfrak{g}$ is semisimple), it can be shown that $S(\mathfrak{g})$ is Noetherian \cite[Chapter 2]{MR1451138}. Since $A$ is assumed to be a reductive subgroup of $G$, the invariant Poisson subalgebra $S(\mathfrak{g})^\mathfrak{a}$ is also finitely-generated by the Hilbert-Nagata theorem \cite[Theorem 6.1]{Dolgachev03}. This implies that, once the set consisting of all the indecomposable polynomials $\left\{p^{(k_1)},\dots ,p^{(k_n)}\right\}$ has been constructed via the polynomial ansatz provided later, there always exists some integer $l \in \mathbb{N}$ such that $p^{(l+k_j)}$ is \textit{decomposable} for all $j \geq 1$. We say that a polynomial $p \in S(\mathfrak{g})$ is decomposable if there exists another polynomial $p' \in S(\mathfrak{g})$ of a lower degree such that $p'$ is a divisor of $p$. Additionally, it is important to note that the elements in the generating set within the centralizer of a subalgebra do not necessarily imply their algebraic independence. See further explanation in Remark \ref{rem:generators}. 
\end{remark}

 Recall that $S(\mathfrak{g}) \cong \mathbb{C}[\mathfrak{g}^*]$. Under this explicit isomorphism, $S(\mathfrak{g})$ inherits the standard grading by polynomial degree  \begin{align}
   S(\mathfrak{g}) = \bigoplus_{k \geq 0} S^k(\mathfrak{g})  ,
\end{align} where  \begin{align}
     S^k(\mathfrak{g}) : = \mathrm{span} \left\{x_1^{a_1} \cdots x_n^{a_n}: a_1 +  \ldots +  a_n = k, \quad a_j \in \mathbb{N}_0 : =\mathbb{N} \cup \{0\}\right\} \label{eq:homopoly}
 \end{align} is a subalgebra of $S(\mathfrak{g})$ consisting of all homogeneous degree $k$ polynomials. It follows that, for any $p \in S(\mathfrak{g})$, the polynomial decomposes as $p = \sum_{k \geq 0} p^{(k)}$, where $p^{(k)} \in S^k(\mathfrak{g})$ for all $k \geq 0$.  We now focus on the $\mathfrak{a}$-invariant homogeneous polynomial subspaces of $\left(S(\mathfrak{g}),\{\cdot,\cdot\}\right)$.   For similar constructions, we refer to \cite{MR191995,MR1520346, MarquetteZhangZhangAOP2025} and the citations therein. Define the vector space of $\mathfrak{a}$-invariant degree $k$-homogeneous polynomials as
\begin{align*}
 S^k(\mathfrak{g})^{\mathfrak{a}}  = \left\{p^{(k)} \in S^k(\mathfrak{g}): \left\{x,p^{(k)}\right\} = 0  \quad \text{ for all } x \in \mathfrak{a}^*\right\},
\end{align*}
where $p^{(k)}(x_1,\ldots,x_n)$ is a homogeneous polynomial of degree $k \geq 0$ with the generic form
\begin{align}
    p^{(k)}(x_1,\ldots,x_n) = \sum_{i_1 +  \ldots +  i_n = k} \Gamma_{i_1,\ldots, i_n}\, x_1^{i_1} \cdots x_n^{i_n}, \qquad  \Gamma_{i_1,\ldots,  i_n} \in \mathbb{C}. \label{eq:ci}
\end{align} 
To find a finite generating set for centralizer subalgebras, all $\mathfrak{a}$-invariant linearly independent and indecomposable homogeneous polynomial solutions of the system of partial differential equations (PDEs)
\begin{align}
    \tilde{X}_m\left(p^{(k)}\right)(x_1,\ldots,x_n) = \left\{x_m,p^{(k)}\right\}  = \sum_{1 \leq l,i \leq n} C_{mi}^lx_l \dfrac{\partial p^{(k)}}{\partial x_i} = 0, \text{ }\quad m = \ell_1,\ldots,\ell_s  \label{eq:func}
\end{align}
must be found.

Recall that if $A\leq G$ is any Lie subgroup acting on the coadjoint representation, under a fixed $\mathrm{Ad}(G)$-invariant bilinear form $B$ with $\mathfrak{g}^* \cong \mathfrak{g}$, we define:  \begin{align}
  S(\mathfrak{g})^A: = \bigl\{f \in S(\mathfrak{g}): f(A \cdot X) = f(X), \text{ for any } X \in \mathfrak{g}\bigr\}, \label{eq:sga}  
\end{align} which is also a Poisson subalgebra: Indeed, $\{f,g\}$ is $A$-invariant when $f,g$ are $A$-invariant. Throughout this work, when we mention $A$-invariant, we refer to $\mathrm{Ad}(A)$-invariant. 
From the previous discussion, we can directly state the following proposition. It demonstrates that determining a generating set for $S(\mathfrak{g})^{\mathfrak{a}}$ is, in fact, equivalent to performing computations within the Poisson algebra $S(\mathfrak{g})^A$. 

\begin{proposition}
\label{prop:equ}
Let $G$ be a $n$-dimensional semisimple Lie algebra $\mathfrak{g}$ with a Killing form $B$. 
 Let $A\subset G$ be a connected closed subgroup with a Lie algebra $\mathfrak{a}\subset\mathfrak{g}$. Choose the basis of $\mathfrak{a}$ to be $\{X_1,\dots,X_s\}$ for some $s\leq n$. For any $f\in S(\mathfrak{g})$, define the adjoint action of $A$ on polynomials by
\begin{align*}
(a \cdot  f)(X):=f\big(aXa^{-1}\big),\qquad a\in A .
\end{align*}
Then for all $X\in\mathfrak{g},\ H\in\mathfrak{a}$ and $f \in S(\mathfrak{g})^A$, the following arguments are equivalent:
\begin{align*}
\text{\emph{(i)}}\quad & df_X \big([H,X]\big)=0\ ;\\
\text{\emph{(ii)}}\quad &
\sum_{i,k=1}^n  C_{jk}^i x_k \frac{\partial f}{\partial x_i} =0
\ \ \text{for each }j=1,\dots,s;\\
\text{\emph{(iii)}}\quad & B \big(\nabla f(X),[H,X]\big)=0\ ,
\end{align*} where $\nabla f$ is the gradient characterised by $df_X(V)=B(\nabla f(X),V)$, and for all $V\in\mathfrak{g}$ and $C_{jk}^i$ are structure constant. Equivalently, define
\begin{align*}
\mathcal{L}_j:=\sum_{i=1}^n\sum_{k=1}^n C_{jk}^i\,x_k\,\frac{\partial}{\partial x_i}\quad  j=1,\dots,s ,
\end{align*} we have \begin{align}
S(\mathfrak{g})^A=S(\mathfrak{g})^{\mathfrak{a}}=\bigcap_{j=1}^s \ker(\mathcal{L}_j). \label{eq:pdefin}
\end{align}
\end{proposition}

\begin{proof}
For any $f \in S(\mathfrak{g})^A$, by definition in \eqref{eq:sga}, $f(\mathrm{Ad}(a^{-1})X) = f(a \cdot X) = f(X)$.  For $H\in\mathfrak{a}$ and the one-parameter subgroup $a(t)=\exp(tH)$, $t \in \mathbb{C}$, we now consider the flow generated by $H$ as follows, \begin{align*}
0=\left.\frac{d}{dt}\right\vert_{t = 0}(a(t) \cdot  f)(X) =\left.\frac{d}{dt}\right\vert_{t = 0}f\big(\mathrm{Ad}(\exp(-tH))X\big)=- df_X \big([H,X]\big).
\end{align*}  Conversely, let $a\in A$ and choose any smooth path $h:[0,1]\to A$ with $h(0)=e$ and $h(1)=a$, where $e \in A$ is the identity element. Define the flow $\Phi(t):=f(\mathrm{Ad}(h(t)^{-1})X)$ for $t \in \mathbb{C}$. Then using chain rules, we have
\begin{align*}
\frac{d}{dt}\Phi(t) =- df_{\Phi(t)} \Big(\left[\dot{h}(t)\,h(t)^{-1},\,\mathrm{Ad}(h(t)^{-1})X\right]\Big)=0,
\end{align*} since $\dot{h}(t)h(t)^{-1}\in\mathfrak{a}$ and using condition (i) for all $H\in\mathfrak{a}$. Hence, $\Phi$ is constant, and $f(\mathrm{Ad}(a^{-1})X)=f(X)$, i.e., $f\in S(\mathfrak{g})^A$.

Now, suppose that (i) holds. By the definition of the $B$-gradient, we have $df_X([H,X]) = B ( \nabla f(X),$ $[H,X] )=0$. Thus, (iii) is concluded.  Finally, we show that $\sum_{i,k=1}^n  C_{jk}^i x_k \frac{\partial f}{\partial x_i}=0 $. For any $H \in \mathfrak{a}$, write $H=\sum_{j=1}^s h_j E_j$ and $X=\sum_{k=1}^n x_k E_k$. Then a direct computation gives \begin{align*}
[H,X]=\sum_{j=1}^s\sum_{k=1}^n h_j C_{jk}^i x_k  X_i,\qquad \nabla f=\sum_{i=1}^n\frac{\partial f}{\partial x_i}\,X_i,
\end{align*} and the coefficients $h_j$ are arbitrary, so each PDE displayed must vanish. Hence, \begin{align*}
    d f_X([H,X]) =  \sum_{i,j,k=1}^{s,n} C_{jk}^i h_j x_k \dfrac{\partial f}{\partial x_i} = \sum_j \mathcal{L}_j(f) = 0.
\end{align*} This gives \eqref{eq:pdefin}, and the proof is completed.
\end{proof}

\begin{remark} 
\label{rmk:property}
(i) In the case where $A$ is not connected. Let $A^\circ$ be the identity component of $A$. The same argument shows $ S(\mathfrak{g})^{A^\circ}=S(\mathfrak{g})^{\mathfrak{a}} = S(\mathfrak{g})^A$, and hence \begin{align*}
S(\mathfrak{g})^A=\bigl(S(\mathfrak{g})^{\mathfrak{a}}\bigr)^{\,A/A^\circ}.
\end{align*}

(ii)  Note that it is shown that if $\mathfrak{a} = \mathfrak{g}$, the maximal number of functionally independent solutions of \eqref{eq:func} is known to be given by \cite{MR0411412,MR0204094}
\begin{align}
    \mathcal{N} (\mathfrak{g}) = \dim  \mathfrak{g} - \mathrm{rank}(A_{ij} )_{n \times n} \, , \text{ } 
      \label{eq:constant}
\end{align} where $A_{ij} :=\sum_{l=1}^n C_{ij}^l x_l$ represents the matrix of the commutator table of the Lie algebra $\mathfrak{g}$ over the given basis. Note that \eqref{eq:constant} still holds if $\mathfrak{a}$ is an Abelian subalgebra. In this context, since $\mathfrak{a}$ is a subalgebra of $\mathfrak{g}$, we will consider the labelling problem where functions (not necessarily polynomials) satisfy the system of PDEs \eqref{eq:func}. It can be shown that the number of functionally independent solutions in the system \eqref{eq:func} is exactly \begin{align}
    \mathcal{N}(\mathfrak{a}) = \dim \mathfrak{g} - \dim \mathfrak{a} +  \ell_0,
\end{align} where $\ell_0$ is the number of $\mathfrak{g}$-invariant polynomials in $S(\mathfrak{a} )$  (for more details, see \cite{MR2276736,MR4660510,MR2515551} and \cite[Chapter 12, Section 12.1.5]{gtp} and references therein). 
\end{remark}

Finally, we focus on the rank of $S(\mathfrak{g})^A$. Recall that for a finitely-generated integral domain $D$ over the base field $\mathbb{C}$, the rank of $D$ is the transcendence degree defined by \begin{align}
    \mathrm{rank}  D: = \mathrm{trdeg}  D. \label{eq:rank}
\end{align} Throughout this paper, the rank of a finitely-generated algebra is defined in \eqref{eq:rank}.

Typically, finding a polynomial for the centralizer in relation to a subalgebra can be approached by two methods: solving PDEs systems \eqref{eq:func} directly, using the method of characteristics, or employing a polynomial ansatz. In this case, we apply the polynomial ansatz, as $\mathfrak{g}$ is assumed to be semisimple, and the commutant can be expressed as a polynomial in dual space variables. This simplifies the analysis by solving sets of linear equations. Note that, in the case of non-semisimple Lie algebras, the solution may be expressed as rational or even transcendental functions.  Recall that, by construction, an $\mathfrak{a}$-invariant homogeneous polynomial in $S^k(\mathfrak{g})^\mathfrak{a}$ takes the form of \eqref{eq:ci}. Subsequently, a list of all polynomials is compiled for each degree, and we examine the decomposability. That is, the degree up to which all polynomials can be expressed in terms of polynomials of lower degrees. Without loss of generality, if indecomposability is achieved up to the degree $\zeta$, the set of polynomials that form the commutant is described by

\[ \textbf{q}_1 : = \left\{ p_u^{(1)},\quad u=1,...,l_1 \right\}; \]
\[ \vdots  \]
\[  \textbf{q}_\zeta : = \left\{p_u^{(\zeta)},\quad u=1,...,l_\zeta\right\}, \]
where $p_u^{(k)} $ is an indecomposable $\mathfrak{a}$-invariant homogeneous polynomial of degree $k \in \{1,\ldots,\zeta\}$.

 Let $\textbf{Q}_\zeta :=    \textbf{q}_1 \sqcup \textbf{q}_2 \sqcup \ldots \sqcup \textbf{q}_\zeta $ be a finite set consisting of all indecomposable polynomials up to degree $\zeta$, and let $\textbf{Alg} \left\langle \textbf{Q}_\zeta \right\rangle $ denote the algebra generated by the set $\textbf{Q}_\zeta$. 
  It is clear that $\textbf{Alg}\langle\textbf{Q}_\zeta\rangle$ is infinite-dimensional as a vector space.  Notice that the elements in $\textbf{Q}_\zeta$ are not necessarily algebraically independent. Hence, they are not freely generated, which means that there may exist non-trivial polynomial relations among these generators. Now, for any $p_u^{(l)} \in \textbf{q}_l$ and $p_v^{(\ell)} \in \textbf{q}_\ell$, there exist some coefficients $ \Gamma^{s_1,...,s_r}_{uv}  \in \mathbb{C}$ such that the Lie-Poisson bracket $\{\cdot,\cdot\}  :  \textbf{Alg} \langle  \textbf{Q}_\zeta \rangle \times \textbf{Alg} \langle  \textbf{Q}_\zeta \rangle\rightarrow \textbf{Alg} \langle  \textbf{Q}_\zeta \rangle$ is given by \begin{equation}
     \left\{p_u^{(l)} ,p_v^{(\ell)} \right\}= \sum_{k_1+\ldots +  k_r =\ell+l-1} \Gamma^{s_1,...,s_r}_{uv} p_{s_1}^{(k_1)} \cdots p_{s_r}^{(k_r)}  . \label{eq:comm}
\end{equation} Here $ \Gamma^{s_1,...,s_r}_{uv} \in \mathbb{C}$ and $k_1,\ldots,k_r  \leq \zeta.$ Subsequently, the algebraic structure denoted by $\textbf{Alg}\left\langle \textbf{Q}_\zeta\right\rangle$, when equipped with the Lie-Poisson bracket $\{\cdot,\cdot\}$, in conjunction with additional polynomial relations $ P(\textbf{q}_1,\ldots,\textbf{q}_\zeta) = 0$, constitutes a finitely-generated Poisson algebra. This algebraic framework is thus characteristic of Poisson algebras, ensuring that it is defined by a finite set of generators.  
We now provide a standard algebraic construction for $A$-invariant subalgebra in $S(\mathfrak{g})$. See, for instance, \cite[Chapter 2]{michor1996isometric}. We starting with the following definition.

\begin{definition} 
\label{def:ftal}
Let $\mathbf{Q}_\zeta \subset S(\mathfrak{g})^A$ be a finite set that generates $S(\mathfrak{g})^A$ as a $\mathcal{Z}$-algebra, where $\mathcal{Z} : = \mathcal{Z}\left(S(\mathfrak{g})^A\right)$ is the Poisson center of $S(\mathfrak{g})^A$.  We write $\mathbf{Alg}\langle \mathbf{Q}_\zeta\rangle$ for the $\mathcal{Z}$-subalgebra of $S(\mathfrak{g})^A$ generated by $p_1,\dots,p_\zeta$. In particular, $\mathbf{Alg}\langle \mathbf{Q}_\zeta\rangle = S(\mathfrak{g})^A$.
\end{definition}

\begin{remark}
    Note that $\mathbf{Q}_\zeta$ is the minimal set of algebra generators of $S(\mathfrak{g})^A$; these generators need not be algebraically independent. The maximal number of algebraically independent generators is the transcendence degree (rank) of $S(\mathfrak{g})^A$, and it must be strictly smaller than $\zeta$.
\end{remark}

We now provide an algebraic description on the Poisson centraliser $S(\mathfrak{g})^A$. Based on the construction provided above, we see that $S(\mathfrak{g})^A$ consists of polynomials with additional relations. In the following proposition, we build the polynomial relations as elements of the kernel of a certain evaluation map.

\begin{proposition} 
\label{prop:polyrelation}
Let $\mathbf{Q}_\zeta  $ be the same as defined above, and let $ \mathbb{C}[a_1 ,\dots, a_\zeta]$ be the free commutative polynomial ring\footnote{ The symbols $a_i$ are formal indeterminates with no relations inside the polynomial ring; the images $\mathrm{ev}(a_i)$ are not algebraically nor functionally independent generators.} 
for the variables $a_1,\dots,a_\zeta$. Define the evaluation homomorphism \begin{align*}
    \mathrm{ev}:\mathbb{C}[a_1,\dots,a_\zeta]\longrightarrow S(\mathfrak{g})^A,\qquad \mathrm{ev}(a_i) = p_i   
\end{align*} for all $1 \leq i\leq \zeta$. Then $\mathrm{ev}$ is surjective and \begin{align*}
    I:= \ker(\mathrm{ev}) = \left\{  P(a_1,\dots,a_\zeta)\in \mathbb{C}[a_1,\dots,a_\zeta]: P(p_1,\dots,p_\zeta) = 0\ \text{in }S(\mathfrak{g})^A \right\}.
\end{align*} Consequently, \begin{align}
    \mathbb{C}[a_1,\dots,a_\zeta]\big/ I\cong S(\mathfrak{g})^A. \label{eq:FIT}
\end{align} 
\end{proposition}

\begin{proof}
Since $\mathbf{Q}_\zeta$ generates $S(\mathfrak{g})^A$ as a $\mathcal{Z}$-algebra by Definition \ref{def:ftal}, every element of $S(\mathfrak{g})^A$ is a polynomial in $p_i$ with $\mathcal{Z}$-coefficients. Hence, $\mathrm{ev}$ is surjective. By definition, an element $P(a_1,\dots,a_\zeta)\in \mathbb{C}[a_1,\dots,a_\zeta]$ lies in $\ker(\mathrm{ev})$ if and only if $P(p_1,\dots,p_\zeta)=0$ is in $S(\mathfrak{g})^A$ as polynomial relations contained in $I$. We then show that \eqref{eq:FIT} holds.

 Since $S(\mathfrak{g})^A$ is a Poisson subalgebra, for $i,j$, we choose polynomials $F_{ij}\in \mathbb{C}[a_1,\dots,a_\zeta]$ with $\{p_i,p_j\}=F_{ij}(p_1,\dots,p_\zeta) \in S(\mathfrak{g})^A$ and define \begin{align*}
    \left\{ a_i+I, a_j+I \right\}:=F_{ij}(a_1,\dots,a_\zeta)+I,
\end{align*} extending by $\mathbb{C}$-bilinearity and the Leibniz rule.  If $F_{ij}'$ is another choice with the same value as $(p_1,\dots,p_\zeta)$, then $F_{ij}-F_{ij}'\in I$. Thus, the definition is independent of the choice and hence well-defined. By the First Isomorphism Theorem, there exists a canonical isomorphism of $\mathcal{Z}$-algebras such that $\mathbb{C}[a_1,\dots,a_\zeta]\big/I \cong S(\mathfrak{g})^A$. 
\end{proof}

\begin{remark}
\label{rem:generators}
(i) Fix a finitely generating set  $\mathbf{Q}_\zeta=\{p_1,\dots,p_\zeta\}\subset S(\mathfrak{g})^A$. For any $p_j \in S(\mathfrak{g})^A$, if $S_0(\mathfrak{g})^A=\bigoplus_{j\geq 1}S_j(\mathfrak{g})^A$ consists of non-constant polynomials ($\deg p >0$), then $p_j\notin (S_0(\mathfrak{g})^A)^2$, i.e., the class of $p_j$ is nonzero in the space of indecomposables $$Q(S(\mathfrak{g})^A):=S_0(\mathfrak{g})^A\big/\left(S_0(\mathfrak{g})^A\right)^2,$$ where $\left(S_0(\mathfrak{g})^A\right)^2:= \left\{pq: p,q \in S_0(\mathfrak{g})^A\right\}$ is the ideal consisting of decomposable generators. The elements $p_j$ need \emph{not} be algebraically or functionally independent. They are simply the explicit generators inside $S(\mathfrak{g})^A$.

(ii) The $a_i$ are indeterminates in the polynomial ring  $ \mathbb{C}[a_1,\dots,a_\zeta]$.  They serve only as formal variables for recording polynomial expressions and relations among $p_j$ through the evaluation map $a_i\mapsto p_i$. In particular, $a_i$ do not denote new invariants within $S(\mathfrak{g})^A$, and any functional dependence of $p_j$ is simply reflected by nontrivial polynomial relations when substituted with $a_i\mapsto p_i$.

(iii) In general, a polynomial algebra is a special case of a finitely generated algebra, as it does not contain any polynomial relations. Although $S(\mathfrak{g})^A$ may contain a polynomial relation, it is therefore not necessarily a polynomial ring. However, as we constructed above via polynomial ansatz, the Poisson bracket relations in $S(\mathfrak{g})^A$ are closed in a polynomial way \eqref{eq:comm}, and regarding whether or not it contains polynomial relations, we always call $S(\mathfrak{g})^A$ a Poisson polynomial algebra. Hence, by a Poisson polynomial algebra, we mean a finitely generated Poisson $\mathbb{C}$-algebra whose underlying commutative algebra is a quotient of a polynomial ring by a Poisson ideal.
\end{remark}



\begin{definition}
\label{H}
Let $C_1,\ldots,C_r$ be basic $G$-invariant polynomials (Casimirs) in $S(\mathfrak{g})^G$, where $r=\mathrm{rank}\,\mathfrak{g}$. Let $Z_P\bigl(S(\mathfrak{g})^A\bigr)$ denote the \emph{Poisson center} of $S(\mathfrak{g})^A$, and set \begin{align*}
 t:=\mathrm{trdeg}_{\mathbb{C}}\, Z_P\bigl(S(\mathfrak{g})^A\bigr).
\end{align*} A Hamiltonian with respect to $S(\mathfrak{g})^A$ is any element \begin{align}
\mathcal{H}\in Z_P\bigl(S(\mathfrak{g})^A\bigr), \label{eq:Hamil}
\end{align} for instance any polynomial expression in a choice of $t$ algebraically independent elements of $Z_P\bigl(S(\mathfrak{g})^A\bigr)$ together with (optionally) a linear combination $\sum_{j=1}^r\gamma_j C_j$.
\end{definition}

A superintegrable system can be naturally correlated with the Poisson centralizer $S(\mathfrak{g})^A$ through the following considerations:

\begin{definition}
\label{def:superal}
A Hamiltonian $\mathcal{H}$ is \textit{algebraic (super)integrable} (with respect to $S(\mathfrak{g})^A$) if there exist $t:=\mathrm{trdeg}_{\mathbb{C}}\, Z_P\bigl(S(\mathfrak{g})^A\bigr)$ algebraically independent elements $c_1,\ldots,c_t\in Z_P\bigl(S(\mathfrak{g})^A\bigr)$ such that $\mathcal{H}\in Z_P\bigl(S(\mathfrak{g})^A\bigr)$ and
\begin{align*}
\mathrm{trdeg}_{\mathbb{C}} S(\mathfrak{g})
=
\mathrm{trdeg}_{\mathbb{C}} S(\mathfrak{g})^A
+
\mathrm{trdeg}_{\mathbb{C}} Z_P\bigl(S(\mathfrak{g})^A\bigr),
\end{align*}
so that we have an inclusion chain of Poisson algebras
\begin{align*}
Z_P\bigl(S(\mathfrak{g})^A\bigr)\subset S(\mathfrak{g})^A\subset S(\mathfrak{g}).
\end{align*}
\end{definition}

\begin{remark}
 Depending on the structure of the subalgebra $\mathfrak{a}$, $S(\mathfrak{g})^A$ can contain a sufficiently high number of functionally independent polynomials $p_1,\dots ,p_w$ such that $\mathcal{S}  $, consisting of a Hamiltonian and first integrals $\mathcal{H}, p_1,\dots ,p_w$, forms a superintegrable system, where the maximum number $w = \mathrm{rank}\, S(\mathfrak{g})^A$ of elements commutes with the algebraic Hamiltonian $\mathcal{H}$.
\end{remark}

\subsection{Geometric interpretation on superintegrable systems in \texorpdfstring{$S(\mathfrak{g})^A$}{S(\mathfrak{g})\^{}A}}
\label{subsec:geomi}
In Subsection \ref{subsec:geomi}, we provide the geometric interpretation of the algebraic superintegrable systems defined in Definition \ref{def:superal}. The main terminologies we applied in this subsection can be found in \cite[Chapter 6]{Dolgachev03} or \cite{Brion2010Actions}. Let $G$ be a complex semisimple Lie group with Lie algebra $\mathfrak{g}$, and let $A\subset G$ be a complex semisimple Lie subgroup. We consider the restriction of the adjoint action of $G$ to $A$. Since $\mathbb{C} [\mathfrak{g}^*]=S(\mathfrak{g})$, the affine invariant-theoretic quotient of $\mathfrak{g}^*$ by $A$ is defined by \begin{align} 
\mathfrak{g}^*//A:=\mathrm{Spec}S(\mathfrak{g})^A.  \label{eq:gitquotient} 
\end{align} By  Hilbert-Nagata theorem, $S(\mathfrak{g})^A$ is a finitely generated $\mathbb{C} $-algebra. Define \begin{align*} 
d_A:=\max_{X\in\mathfrak{g}}\dim\bigl(\mathrm{Ad}(A)\cdot X\bigr). 
\end{align*} The maximal orbit dimension is attained on a non-empty Zariski open subset of $\mathfrak{g}$. Hence, by the dimension formula, we have \begin{align} 
\dim(\mathfrak{g}^*//A) =\mathrm{trdeg}_{\mathbb{C}}S(\mathfrak{g})^A = \dim \mathfrak{g} -d_A = \dim G -d_A.  \label{eq:dimgitquo} 
\end{align}

Fix a non-degenerate $\mathrm{Ad}(G)$-invariant complex bilinear form $B$ on $\mathfrak{g}$. Via $B$ we identify $\mathfrak{g}\cong\mathfrak{g}^*$, so that we may view polynomial functions on $\mathfrak{g}$ as the symmetric algebra $S(\mathfrak{g})=\mathbb{C}[\mathfrak{g}]$. For each $X\in\mathfrak{g}$, let
\begin{align*}
\mathrm{ev}_X:S(\mathfrak{g})\to\mathbb{C},\qquad f\longmapsto f(X)
\end{align*}
be the evaluation homomorphism. Similar to \eqref{eq:gitquotient}, we define the affine (GIT) quotient of $\mathfrak{g}$ by $A$ to be $\mathfrak{g}//A:=\mathrm{Spec}\bigl(S(\mathfrak{g})^A\bigr)$. Let $\mathsf{i}_A:S(\mathfrak{g})^A \hookrightarrow S(\mathfrak{g})$ denote the inclusion of invariant polynomials. By functoriality of $\mathrm{Spec}$, this inclusion induces a morphism
\begin{align}
\chi_A: \mathfrak{g}  \longrightarrow   \mathfrak{g}//A, \qquad X\longmapsto \mathrm{ev}_X\vert_{S(\mathfrak{g})^A},  \label{eq:canonicalpro}
\end{align}
whose pullback map on functions is precisely \begin{align*}
\widehat{\chi}_A = \mathsf{i}_A:S(\mathfrak{g})^A\hookrightarrow S(\mathfrak{g}).
\end{align*} In particular, for $X,Y\in\mathfrak{g}$, \begin{align*}
\chi_A(X) = \chi_A(Y) \Longleftrightarrow f(X)=f(Y)\text{ for all }f\in S(\mathfrak{g})^A.
\end{align*} For $X\in\mathfrak{g}$, we denote by \begin{align}
[X]_A :=\chi_A^{-1}\bigl(\chi_A(X)\bigr) =\{Y\in\mathfrak{g}: f(Y)=f(X)\text{ for all }f\in S(\mathfrak{g})^A\} \label{eq:pointsinquo}
\end{align} the fiber of the quotient morphism through $X$. Equivalently, $[X]_A$ is the invariant-theoretic equivalence class of $X$ (two points are equivalent if they cannot be separated by invariant polynomials). In general, this fiber may contain several distinct $A$-orbits, so it need not coincide with the single orbit $A\cdot X$.

Let $f_1, \ldots, f_\zeta$ be homogeneous generators of $S(\mathfrak{g})^A$. By Proposition \ref{prop:polyrelation}, we deduce \begin{align*}
S(\mathfrak{g})^A\simeq\mathbb{C} [t_1,\ldots,t_\zeta]/I,\qquad \mathfrak{g}//A\simeq V(I)\subset\mathbb{A} ^\zeta. 
\end{align*} Under this identification, the quotient morphism is the map\begin{align*}
\chi_A:\mathfrak{g}\longrightarrow V(I)\subset\mathbb{A} ^\zeta,\qquad X\longmapsto\bigl(f_1(X),\ldots,f_\zeta(X)\bigr). 
\end{align*} Therefore, for all $X,Y\in\mathfrak{g}$,  we deduce \begin{align*} 
\chi_A(X) = \chi_A(Y)\Longleftrightarrow f_i(X) = f_i(Y)\text{ for every }1 \leq i\leq \zeta \Longleftrightarrow f(X) = f(Y)\text{ for every } f \in S(\mathfrak{g})^A. 
\end{align*} Moreover, \begin{align*} 
\mathfrak{g}//A\simeq\mathbb{A} ^\zeta\Longleftrightarrow I=(0)\Longleftrightarrow f_1,\ldots,f_\zeta\text{ are lineally independent}.
\end{align*}

Since $A$ is reductive (in particular, semisimple), the invariant ring $S(\mathfrak{g})^A$ is finitely generated and the morphism $\chi_A:\mathfrak{g}\to\mathfrak{g}//A=\mathrm{Spec}\bigl(S(\mathfrak{g})^A\bigr)$ is the affine GIT quotient (hence a good quotient). Consequently, every fiber of $\chi_A$ contains a unique closed $A$-orbit. Equivalently, \begin{align} 
\chi_A(X) = \chi_A(Y)\Longleftrightarrow \overline{A\cdot X}\cap\overline{A\cdot Y}\neq\varnothing. \label{eq:relations}
\end{align} In particular, the closed $A$-orbits in $\mathfrak{g}$ are in bijection with the closed points of $\mathfrak{g}//A$. If $A\cdot X$ and $A\cdot Y$ are both closed, then this relation \eqref{eq:relations} reduces to $\chi_A(X)= \chi_A(Y)$ if and only if $A\cdot X = A\cdot Y$. 

We now introduce the regular locus, which enable us to define smooth morphism on the GIT quotient.
\begin{definition} 
\label{def:regularstr} \cite[Chapter 1, Section 5]{HartshorneAG}, \cite{stacks01V4}
The smooth value locus of the quotient morphism $\chi_A$ is the maximal Zariski open subset of $\mathfrak{g} //A$ over which $\chi_A$ is smooth: \begin{align}
    (\mathfrak{g} //A)_{\mathrm{reg}}  := \bigcup    \left\{   \mathcal{U} \subset \mathfrak{g} //A :   \mathcal{U} \text{ is Zariski open and }     \chi_A^{-1}(\mathcal{U})\to \mathcal{U} \text{ is smooth}  \right\}.
\end{align} Equivalently, $y\in(\mathfrak{g} //A)_{\mathrm{reg}}$ if and only if there exists a Zariski open neighbourhood $y\in \mathcal{U}\subset \mathfrak{g} //A$ such that \begin{align*}
    \chi_A^{-1}(\mathcal{U})\longrightarrow \mathcal{U}
\end{align*} is a smooth morphism.
\end{definition}

\begin{proposition} 
\label{prop:smoothregularg} \cite[Proposition 1.26]{Brion2010Actions}
The affine quotient $\mathfrak{g}//A$ carries a unique Poisson structure for which $\chi_A:\mathfrak{g}\to\mathfrak{g}//A$ is a Poisson morphism. Moreover, $(\mathfrak{g}//A)_{\mathrm{reg}}$ is a non-empty Zariski open dense subset of $\mathfrak{g}//A$, and the restricted morphism \begin{align*} \chi_A:\chi_A^{-1}\bigl((\mathfrak{g}//A)_{\mathrm{reg}}\bigr)\longrightarrow(\mathfrak{g}//A)_{\mathrm{reg}} \end{align*} is a smooth Poisson morphism. In particular, after passing to the associated complex manifolds, this restriction is a Poisson submersion.
\end{proposition}

\begin{proof}
Since the adjoint action of $A$ preserves the Lie bracket on $\mathfrak{g}$, it also preserves the Lie-Poisson bracket on $S(\mathfrak{g})$. Hence, for $f,h\in S(\mathfrak{g})^A$ and $a\in A$, one has \begin{align*} 
a\cdot\{f,h\}=\{a\cdot f,a\cdot h\}=\{f,h\}.
\end{align*} Thus, $S(\mathfrak{g})^A$ is a Poisson subalgebra of $S(\mathfrak{g})$. This gives a Poisson structure on \begin{align*}
\mathbb{C} [\mathfrak{g}//A] = S(\mathfrak{g})^A. 
\end{align*} Since $\widehat{\chi}_A:S(\mathfrak{g})^A \hookrightarrow S(\mathfrak{g})$ is the inclusion, $\chi_A$ is a Poisson morphism. The uniqueness of the Poisson structure follows from the injectivity of $\widehat{\chi}_A$, because the bracket on $\mathbb{C} [\mathfrak{g}//A]$ is forced to satisfy
\begin{align*} \widehat{\chi}_A\{u,v\}_{\mathfrak{g}//A}=\{\widehat{\chi}_Au,\widehat{\chi}_Av\}_{\mathfrak{g}},\qquad u,v\in S(\mathfrak{g})^A. \end{align*}

It remains to prove the smoothness statement. The variety $\mathfrak{g}$ is smooth and irreducible, and the morphism $\chi_A:\mathfrak{g}\to\mathfrak{g}//A$ is dominant because $\widehat{\chi}_A$ is injective. Since the ground field has characteristic zero, by Definition \ref{def:regularstr}, regularity on $\mathfrak{g}//A$ implies that there exists a non-empty Zariski open subset $U\subset\mathfrak{g}//A$ such that
\begin{align*} \chi_A^{-1}(\mathcal{U})\longrightarrow \mathcal{U} \end{align*}
is smooth. By definition of the regular value locus, $U\subset(\mathfrak{g}//A)_{\mathrm{reg}}$. Therefore, $(\mathfrak{g}//A)_{\mathrm{reg}}$ is non-empty, Zariski open, and dense. Since $\chi_A$ is Poisson on the whole quotient, the restricted morphism remains Poisson. Since $\chi_A$ is Poisson on the whole quotient, the restricted morphism remains Poisson. Finally, a smooth morphism of complex Lie varieties induces a submersion on the associated complex manifolds, and hence the restricted morphism is a Poisson submersion.
\end{proof}


We now provide a geometric interpretation of the definition of a superintegrable system in $S(\mathfrak{g})$, which can be seen as the combination of both Definition \ref{def:superal} and Definition \ref{def:superge}. 

\begin{definition}
\label{def:superalge}
Let $\mathcal{X}$ be an affine Poisson variety with coordinate algebra $\mathcal{O}(\mathcal{X})$. The inclusion chain of   finitely generated Poisson subalgebras \begin{align*}
\mathcal{B}\subset \mathcal{A}\subset \mathcal{O}(\mathcal{X})
\end{align*} defines a superintegrable system on $\mathcal{X}$ if \begin{center}
    (a) \  $\mathcal{B}\subset \mathcal{Z}(\mathcal{A})$. That is,  $\mathcal{B}_k$ is in the Poisson center of $\mathcal{A}$;  \quad (b) \  $\mathrm{trdeg} \,\mathcal{A}+\mathrm{trdeg} \,\mathcal{B}=\dim \mathcal{X}$.
\end{center}   Equivalently, the induced chain of affine Poisson morphisms \begin{align*}
\mathcal{X} \xrightarrow{\ \pi\ } \mathrm{Spec} \,\mathcal{A} \xrightarrow{\ \rho\ } \mathrm{Spec}\,\mathcal{B}
\end{align*} has $\dim \mathcal{X}=\dim\mathrm{Spec}\,\mathcal{A}+\dim\mathrm{Spec}\,\mathcal{B}$.
\end{definition}

\begin{remark} \label{rem:regularpoisub}
Definition \ref{def:superalge} illustrates an affine-algebraic expression:  the chain $  X \xrightarrow{\ \pi\ } \mathrm{Spec} \, \mathcal{A} \xrightarrow{\ \rho\ } \mathrm{Spec} \, \mathcal{B} $  is required to be a chain of Poisson morphisms (i.e., the pullbacks $\pi^*,\rho^*$ are Poisson algebra morphisms), together with the dimension identity \begin{align*}
    \dim X=\dim\mathrm{Spec} \,\mathcal{A}+\dim\mathrm{Spec} \, \mathcal{B}  \ \Longleftrightarrow  \mathrm{trdeg} \, \mathcal{A}+\mathrm{trdeg} \,\mathcal{B}=\dim X .
\end{align*}  In particular, $\mathrm{Spec} \,\mathcal{A}$ and $\mathrm{Spec} \,\mathcal{B}$ are \emph{not} assumed to be smooth. 

We can recapture a Poisson submersion chain defined in Definition \ref{def:superge} by restricting to a Zariski open dense \emph{regular locus} on which the relevant Jacobian ranks are maximal. See Definition \ref{def:regularstr}. Concretely, for a morphism $f\colon X\to Y$ with $X$ smooth, let $Y_{\mathrm{reg}}(f)\subset Y$ denote the smooth value locus of $f$ in the sense of Definition \ref{def:regularstr}, and define \begin{align*}
  X_{\mathrm{reg}}(f):=f^{-1}\bigl(Y_{\mathrm{reg}}(f)\bigr)\subset X.
\end{align*}
Equivalently (since $X$ is smooth), $X_{\mathrm{reg}}(f)$ is the Zariski open locus on which the differential $df_x$ has maximal rank. If $Y$ is singular, the phrase \emph{$f$ is a (Poisson) submersion} is therefore to be understood after restricting to $X_{\mathrm{reg}}(f)\to Y_{\mathrm{reg}}(f)$, where the target is smooth and the restriction is a smooth morphism (hence a submersion on associated complex manifolds).


On such regular loci, the induced maps in the corresponding chains become Poisson submersions between smooth Poisson manifolds, such that the geometric hypotheses of Definition \ref{def:superge} apply without modifying the algebraic superintegrability of Definition \ref{def:superalge} as restricting to a Zariski open dense subset does not change the function-field
level relations.
\end{remark}

Let $\chi_A:\mathfrak{g} \rightarrow \mathfrak{g}//A $ be the categorical quotient map defined above. Since $\chi_A$ collapses each $A$-orbit, we have the fundamental identity $ \dim \mathfrak{g}//A  =\dim\mathfrak{g} - d_A $. By Definition \ref{def:superalge}, to construct a superintegrable system, we seek a chain of affine Poisson varieties \begin{align}
     \mathfrak{g} =\mathrm{Spec}\, S(\mathfrak{g}) \,\xrightarrow{\,\chi_A\,}\, \mathfrak{g}//A =\mathrm{Spec} \, S(\mathfrak{g})^A \,\xrightarrow{\,\rho\,}\, \mathrm{Spec} \, \mathcal{B},
\end{align} where $\mathcal{B}\subset S(\mathfrak{g})^A$ is a finitely-generated Poisson-commutative polynomial ring of transcendence degree $m$, such that $
\dim\mathfrak{g} =\dim \mathfrak{g}//A+m$. In what follows, we will concentrate on a specific case, letting $A = T$.

Let $T\subset G$ be a maximal torus with $\dim T=r$. Then $d_T=r$, and Chevalley’s restriction theorem produces \begin{align}
  S(\mathfrak{g})^G \cong S(\mathfrak{t})^W \cong \mathbb{C}[F_1,\ldots,F_r], \label{eq:crt}
\end{align} where $F_1,\ldots,F_r$ are algebraically independent homogeneous generators. Hence $S(\mathfrak{g})^G$ is a polynomial ring and therefore $\mathfrak{g}//G \cong \mathbb{A}^r$.  The Kostant theorem identifies $S(\mathfrak{g})^G$ with the Poisson center $\mathcal{Z}\bigl(S(\mathfrak{g})\bigr)$, and since $T \subset G$, we must have $S(\mathfrak{g})^G \subset S(\mathfrak{g})^T$.  
We then claim that the chain \begin{align}
    \mathrm{Spec}\, S(\mathfrak{g}) \,\xrightarrow{\ \chi_T\ }\,
 \mathrm{Spec}\, S(\mathfrak{g})^T \,\xrightarrow{\ \rho\ }\,
\mathrm{Spec}\, S(\mathfrak{g})^G
\end{align} is superintegrable. 

\begin{theorem} 
\label{thm:superalchain}
Let $G$ be a semisimple and connected Lie group over a field $\mathbb{C}$ with Lie algebra $\mathfrak{g}$, and let $T\subset G$ be a maximal rank torus $r=\dim T$.  Then the inclusion of Poisson algebras\begin{align*}
S(\mathfrak{g})^G  \subset  S(\mathfrak{g})^T  \subset  S(\mathfrak{g})
\end{align*}makes the chain of affine Poisson varieties \begin{equation} 
\mathfrak{g}\xrightarrow{\,\chi_T\,}\mathfrak{g}//T \xrightarrow{\,\rho\,} \mathfrak{g}//G      \label{eq:chain}
\end{equation} a superintegrable system. More precisely:
\begin{align*}
\text{(i)}\quad & S(\mathfrak{g})^G \subset \mathcal{Z}\big(S(\mathfrak{g})^T\big),\\
\text{(ii)}\quad & \mathrm{trdeg} \, S(\mathfrak{g})^G=r,\qquad
\mathrm{trdeg} \, S(\mathfrak{g})^T=\dim\mathfrak{g} - r,\\
\text{(iii)}\quad & \dim\mathrm{Spec} \, S(\mathfrak{g}) =\bigl(\dim\mathfrak{g}-r\bigr)+r  .
\end{align*}
\end{theorem}

\begin{proof}
We start with (i). From the discussion above, we know $S(\mathfrak{g})^G \subset S(\mathfrak{g})^T$. Now, for any $f \in S(\mathfrak{g})^G$, we have $f \in \mathcal{Z}(S(\mathfrak{g}))$. As $S(\mathfrak{g})^T \subset S(\mathfrak{g})$, for any $g \in S(\mathfrak{g})^T$, we must have $\{f,g\} = 0$. Hence (i) is claimed.  Moreover, using \eqref{eq:crt}, we have $\mathrm{trdeg} \, S(\mathfrak{g})^G=r$, proving the first identity in (ii). By the definition of the spectrum, we identify $\mathfrak{g} =\mathrm{Spec} \, S(\mathfrak{g})$.

To conclude that \eqref{eq:chain} is superintegrable, we only need to check the dimension.  Consider the $T$-action on $\mathfrak{g} $ by the coadjoint representation, which has a generic orbit dimension $r$. Indeed, $\mathfrak{g}$ is decomposed into $T$-weight spaces \begin{align*}
\mathfrak{g}=\mathfrak{t} \oplus\bigoplus_{\alpha\in\Phi}\mathfrak{g}_\alpha,
\qquad \left.\mathrm{Ad}(t)\right\vert_{\mathfrak{g}_\alpha}=\alpha(t)\cdot \mathrm{id}.
\end{align*}  Here, $\Phi$ is the set of roots for $\mathfrak{g}$. Similar decomposition holds for $\mathfrak{g}^*$. 
The stabiliser in $T$ is then finite, so $\dim T\cdot X=r$ with $X \in \mathfrak{g}$. By the dimension formula for the GIT quotients in \eqref{eq:dimgitquo}, \begin{align*}
\mathrm{trdeg} \, S(\mathfrak{g})^T=\dim \mathfrak{g} - \dim(\text{generic }T\text{-orbit}) =\dim\mathfrak{g} - r.
\end{align*} Hence, (ii) holds. Since $\dim\mathrm{Spec} \, S(\mathfrak{g})=\dim\mathfrak{g}$, together with (ii), gives \begin{align*}
\dim\mathfrak{g} = \bigl(\dim\mathfrak{g}-r\bigr)+r = \dim\mathrm{Spec} \, S(\mathfrak{g})^T+\dim\mathrm{Spec} \, S(\mathfrak{g})^G.
\end{align*}
As $S(\mathfrak{g})^G\subset S(\mathfrak{g})^T\subset S(\mathfrak{g})$ are Poisson subalgebras and $S(\mathfrak{g})^G$ lies in the Poisson center of $S(\mathfrak{g})^T$, the induced maps in \eqref{eq:chain} are Poisson morphisms, and the pair
\begin{align}
\mathcal{B} =S(\mathfrak{g})^G   \subset  \mathcal{A} =S(\mathfrak{g})^T
\end{align}
defines a superintegrable system in the sense of Definition \ref{def:superalge}. This proves (iii) and the claim.  
\end{proof}

Let $\chi:\mathfrak{g}\to \mathfrak{g}//G$ be the canonical quotient map and let
\begin{align*}
\chi_T:\mathfrak{g} \longrightarrow  \mathfrak{g}//T ,\qquad \rho:\mathfrak{g}//T \longrightarrow  \mathfrak{g}//G
\end{align*}
be induced by the inclusions $S(\mathfrak{g})^G\hookrightarrow S(\mathfrak{g})^T \hookrightarrow S(\mathfrak{g})$, so that $\chi = \rho\circ \chi_T$. 

In what follows, let $M = G/A$. Recall that in \cite{jiang2026poisson}, we constructed the superintegrable Poisson chain on the cotangent bundle of a homogeneous space with the magnetic symplectic form. We will use the map \begin{align*}
P:T^*M\longrightarrow \mathfrak{g},\qquad  P(g,X):=\mathrm{Ad}(g)(X),
\end{align*} which is the moment map in our trivialisation of $T^*M$. We then consider the regular superintegrable system chain. That is, consider the case with $A = T$.

The vector space $\mathfrak{m} = T_{eA}(G/T)$ is $T$-invariant. Define
\begin{align*}
\overline \chi_T:\mathfrak{m}//T\ \longrightarrow\ \mathfrak{g}//G,\qquad [u]_T\longmapsto \chi(u),
\end{align*}
where $[u]_T$ denotes the $T$-orbit class of $u\in \mathfrak{m}$. The corresponding fiber product is
\begin{align*}
Y := \mathfrak{g} \times_{\mathfrak{g}//G}\mathfrak{m}//T
= \{(X,[u]_T)\, :\, \chi(X)=\chi(u)\}.
\end{align*}
We equip $Y$ with the unique Poisson structure for which the projections $\mathrm{pr}_1:Y\to \mathfrak{g}$ and $\pi_2:Y\to \mathfrak{g}//G$ are Poisson.

\begin{theorem}
\label{thm:spectralequivalent}
Consider the superintegrable system chain
\begin{align*}
T^*M \xrightarrow{\ \pi_1\ }Y \xrightarrow{\ \pi_2\ } \mathfrak{g}//G, \qquad \begin{matrix}
    \pi_1(g,X):=\left(P(g,X),[u]_T\right), \\
    \pi_2\left(P(g,X),[u]_T\right) : = (C_1,\ldots,C_r),
\end{matrix}
\end{align*}
where $(C_1,\ldots,C_r)$ are the basic invariant polynomials (i.e., coordinate functions on $\mathfrak{g}//G$) evaluated at the $\mathfrak{g}$-component.

Together with the chain $\mathfrak{g} \xrightarrow{\ \chi_T\ } \mathfrak{g}//T \xrightarrow{\ \rho\ } \mathfrak{g}//G$, these two systems are spectrally equivalent in the sense of Definition \ref{def:equidefi}. In particular, we have the commutative diagram
\[\begin{tikzcd}
	{T^*M} & {\mathfrak{g} \times_{\mathfrak{g}//G} (\mathfrak{m} - \varepsilon W)//T} & {\mathfrak{g}//G} \\
	{\mathfrak{g}^* \cong \mathfrak{g}} & {\mathfrak{g}//T=\mathrm{Spec} \, S(\mathfrak{g})^T} & {\mathfrak{g}//G}
	\arrow["{\pi_1}", from=1-1, to=1-2]
	\arrow["{P}"', from=1-1, to=2-1]
	\arrow["{\pi_2}", from=1-2, to=1-3]
	\arrow["{\phi_1}"', from=1-2, to=2-2]
	\arrow["{\phi_2}"', from=1-3, to=2-3]
	\arrow["{\chi_T}"', from=2-1, to=2-2]
	\arrow["\rho"', from=2-2, to=2-3]
    \arrow[from=1-1, to=2-2, phantom, "\circlearrowright" {anchor=center, scale=1.5, rotate=90}]
    \arrow[from=1-2, to=2-3, phantom, "\circlearrowright" {anchor=center, scale=1.5, rotate=90}].
\end{tikzcd}\]
This equivalence is implemented by the Poisson maps \begin{align}
\phi:=P:T^* M \rightarrow  \mathfrak{g}^*,\qquad \phi_1:=\chi_T\circ \mathrm{pr}_1:Y  \rightarrow  \mathfrak{g}//T,\qquad \phi_2:=\mathrm{id}_{\mathfrak{g}//G}.  \label{eq:mp}
\end{align}
They satisfy the commutation relations
\begin{align*}
\chi_T\circ \phi=\phi_1\circ \pi_1,\qquad \rho\circ \phi_1=\pi_2.
\end{align*}
\end{theorem}

\begin{proof}
By Theorem \ref{thm:superalchain}, these two chains are superintegrable.   From the assumption above, we have
$\phi =P$, $\phi_1:=\chi_T\circ \mathrm{pr}_1$, and $\phi_2:=\mathrm{id}_{\mathfrak{g}//G}$. It is clear that all the morphisms defined in \eqref{eq:mp} are Poisson. Indeed, $\phi$ is a moment map, $\phi_1$ as a composition of Poisson maps, and $\phi_2 =\mathrm{id}$ is the identity map and hence Poisson. For $(g,X)\in T^* M$, \begin{align*}
\chi\bigl(P(g,X)\bigr) =\chi\bigl(\mathrm{Ad}(g)(u-\varepsilon W)\bigr) =\chi(u-\varepsilon W) = \overline\chi_T([u]_T),
\end{align*} so $\pi_1$ is well defined with $\mathrm{pr}_1\circ\pi_1=P$.  Hence, \begin{align*}
\phi_1\circ \pi_1 =(\chi_T\circ \mathrm{pr}_1)\circ \pi_1 = \chi_T\circ P =\chi_T\circ \phi,
\end{align*} gives the communication of the left part of the diagram. In $Y$, we have $\chi\circ \mathrm{pr}_1=\pi_2$ by the definition of the fiber product, while $\chi=\rho\circ \chi_T$ in $\mathfrak{g}$. Therefore, \begin{align*}
\rho\circ \phi_1 =\rho\circ \chi_T\circ \mathrm{pr}_1 = \chi\circ \mathrm{pr}_1 =\pi_2,
\end{align*} gives the communication of the right part of the diagram. Thus, the diagram is commutative. By Definition \ref{def:equidefi}, these two systems are spectrally equivalent.
\end{proof}


\subsection{The superintegrabilities chains for the case \texorpdfstring{$A\neq T$}{A≠T} }
\label{subsec:AnotequalT}
In Theorem \ref{thm:superalchain},  we take $A = T$, and form the superintegrability. In particular, we showed that the base algebra is chosen to be the Poisson center $S(\mathfrak{g})^G$. For a general closed reductive subgroup $A \subset G$, the base algebra is neither unique nor canonically equal to $S(\mathfrak{g})^G$. In Subsection \ref{subsec:AnotequalT}, we aim to determine possible finitely generated base algebras that give rise to a superintegrable Poisson inclusion chain with the fixed intermediate Poisson algebra $S(\mathfrak{g})^A$.

Let $\mathfrak{a} = \mathrm{Lie} (A)$. Recall that \begin{align*}
    d_A: = \max_{X\in g}\dim\bigl(\mathrm{Ad}(A)\cdot X\bigr).
\end{align*} By \eqref{eq:dimgitquo}, we find $ \mathrm{trdeg} S(\mathfrak{g})^A = \dim \mathfrak{g} - d_A$. Since $Z(S(\mathfrak{g})^A)$ is the Poisson center of $S(\mathfrak{g})^A$, every subalgebra of $Z(S(\mathfrak{g})^A)$ is automatically a Poisson commutative subalgebra. 

\begin{proposition}
\label{prop:superinchainA}
    Let $A \subset G$ be a reductive subgroup of $G$, and let  $  \rho_B: \mathfrak{g}//A = \mathrm{Spec} \, S(\mathfrak{g})^A \longrightarrow \mathrm{Spec} \, \mathcal{B}$  be the morphism induced by the inclusion $\mathcal{B} \subset S(\mathfrak{g})^A$. Then the chain of affine Poisson varieties \begin{align}
        \mathfrak{g} \xrightarrow{ \ \chi_A \ } \mathfrak{g}//A \xrightarrow{ \ \rho_B \ } \mathrm{Spec }\, \mathcal{B} \label{eq:superintechaingA}
    \end{align} is superintegrable in the sense of Definition \ref{def:superalge} if and only if $\mathrm{trdeg} \ \mathcal{B}  = d_A$. Consequently, the set of all possible base algebras for superintegrable chains with the fixed intermediate Poisson algebra $S(\mathfrak{g})^A$ is \begin{align*}
        \mathcal{B}_A:=\left\{B\subset Z \left(S(\mathfrak{g})^A\right)\ : \ \mathcal{B} \text{ is a finitely generated Poisson subalgebra and } \mathrm{trdeg} \ \mathcal{B} =d_A\right\}.
    \end{align*}
\end{proposition}

\begin{proof}
We first assume that the chain $ \mathfrak{g} \xrightarrow{ \ \chi_A \ } \mathfrak{g}//A \xrightarrow{ \ \rho_B \ } \mathrm{Spec}\, \mathcal{B}$ is superintegrable. By Definition \ref{def:superalge}, we have \begin{align*}
    \mathcal{B} \subset Z(S(\mathfrak{g})^A) \quad \text{ and } \quad \dim \mathfrak{g} = \mathrm{trdeg} \, S(\mathfrak{g})^A + \mathrm{trdeg} \ \mathcal{B} .
\end{align*} Using \eqref{eq:dimgitquo}, we then have $\dim \mathfrak{g} - d_A + \mathrm{trdeg}\, \mathcal{B} = \dim \mathfrak{g}$. Hence, $\mathrm{trdeg}\, \mathcal{B} = d_A$.

On the other hand, assume that $\mathcal{B} \subset Z(S(\mathfrak{g})^A)$ is a finitely generated Poisson commutative algebra satisfying $\mathrm{trdeg} \, \mathcal{B} = d_A$. Since $\mathcal{B}$ lies in the Poisson center $Z(S(\mathfrak{g})^A)$, the condition (a) in Definition \ref{def:superalge} is automatically satisfied. Therefore, \begin{align*}
    \mathfrak{g} \xrightarrow{ \ \chi_A \ } \mathfrak{g}//A \xrightarrow{ \ \rho_B \ } \mathrm{Spec}\, \mathcal{B}
\end{align*} is superintegrable, and the description of $\mathcal{B}_A$ follows immediately.
\end{proof}
\begin{remark}
    \label{rem:polybaseA}
(i) Suppose that there exist algebraically independent generators $b_1,\ldots,b_{d_A}$. Then base algebras are polynomial algebras of the form $\mathcal{B} = \mathbb{C}[b_1,\ldots,b_{d_A}]$, and $\mathrm{trdeg} \ \mathcal{B}  = d_A$. By Proposition \ref{prop:superinchainA}, $\mathcal{B}$ is the base algebra of superintegrable chain \eqref{eq:superintechaingA}.

(ii) Then the following conditions are equivalent:
\begin{enumerate}
\item there exists a superintegrable chain of affine Poisson varieties $\mathfrak{g} \xrightarrow{\chi_A} \mathfrak{g}//A \xrightarrow{\rho_B} \mathrm{Spec}\,\mathcal{B}$;
\item there exists a finitely generated Poisson subalgebra $ \mathcal{B}\subset Z(S(\mathfrak{g})^A)$ such that $\mathrm{trdeg} \ \mathcal{B}  = d_A$; 
\item we have  $\mathrm{trdeg}\, Z(S(\mathfrak{g})^A) \geq d_A$. 
\end{enumerate}
If these conditions hold, then one may choose algebraically independent elements $b_1, \dots, b_{d_A} \in Z(S(\mathfrak{g})^A)$ and take \begin{align*}
\mathcal{B} = \mathbb{C}[b_1,\dots,b_{d_A}].
\end{align*}
In particular, the base algebra can always be chosen to be a polynomial Poisson-commutative algebra.
\end{remark}

\begin{corollary}
    \label{coro:nobasechain}
    Suppose that $\mathrm{trdeg}\, Z(S(\mathfrak{g})^A) \leq d_A$. Then, with the fixed intermediate algebra $S(\mathfrak{g})^A$, there does not exist a superintegrable chain in \eqref{eq:superintechaingA}. 
\end{corollary}
\begin{remark}
    Theorem \ref{thm:superalchain} shows that the torus case is special: when $A = T$, we may take \begin{align*}
\mathcal{B}  =  S(\mathfrak{g})^G\subset Z \left(S(\mathfrak{g})^T\right),
\end{align*} and then $\mathrm{trdeg} \ \mathcal{B}  = \mathrm{rank} \, \mathfrak{g}  =  d_T$. For a general subgroup $A$, however, no such canonical choice is available unless we first prove that $Z(S(\mathfrak{g})^A)$ has a transcendence degree of at least $d_A$.
\end{remark}

Now, let us consider another special case of $A$. In this case, we choose $A$ to be the normaliser of $T$. In the following proposition, we will show that the base space in this case is still the Poisson center.

\begin{proposition}\label{prop:normalizerchain}
Let $T\subset G$ be a fixed maximal torus, and let \begin{align*}
    N_G(T):=\{g\in G\mid gTg^{-1}=T\}
\end{align*} denote its normalizer. Set $A:=N_G(T)$ and let $W:=N_G(T)/T$ be the corresponding Weyl group. Then \begin{align*}
S(\mathfrak{g})^A = \left(S(\mathfrak{g})^T\right)^W.
\end{align*} Moreover, \begin{align*}
\mathrm{trdeg}\,S(\mathfrak{g})^A = \mathrm{trdeg}\,S(\mathfrak{g})^T = \dim\mathfrak{g}-\dim\mathfrak{t}, \qquad S(\mathfrak{g})^G\subset Z \left(S(\mathfrak{g})^A\right).
\end{align*} Hence, the chain \begin{align*}
\mathfrak{g} \xrightarrow{\chi_A} \mathfrak{g}//A \xrightarrow{\rho_A} \mathfrak{g}//G
\end{align*} is superintegrable.
\end{proposition}

\begin{proof}
 Since $A = N_G(T)$ and the Weyl group $W = N_G(T)/T$, we have the connected component $A^\circ = T$ and $A/A^\circ \cong W$. Thus Remark \ref{rmk:property} (i) gives \begin{align*}
S(\mathfrak{g})^A = \bigl(S(\mathfrak{g})^{A^\circ}\bigr)^{A/A^\circ} = \left(S(\mathfrak{g})^T\right)^W.
\end{align*} As $W$ is finite, the extension \begin{align*}
\left(S(\mathfrak{g})^T\right)^W \subset S(\mathfrak{g})^T
\end{align*} is integral. Hence \begin{align*}
\mathrm{trdeg} \, S(\mathfrak{g})^A = \mathrm{trdeg}\, S(\mathfrak{g})^T.
\end{align*} 

On the other hand, by Theorem \ref{thm:superalchain}, we deduce $S(\mathfrak{g})^G = Z \left(S(\mathfrak{g})\right)$, whence  $S(\mathfrak{g})^G\subset Z \left(S(\mathfrak{g})^A\right)$. Checking the dimension leads to \begin{align*}
\mathrm{trdeg}S(\mathfrak{g})^A+\mathrm{trdeg}S(\mathfrak{g})^G = \mathrm{trdeg}S(\mathfrak{g})^T+\mathrm{trdeg}S(\mathfrak{g})^G = \dim \mathfrak{g}.
\end{align*} 
Therefore,  \begin{align*}
\mathfrak{g} \xrightarrow{\chi_A} \mathfrak{g}//A \xrightarrow{\rho_A} \mathfrak{g}//G
\end{align*} is superintegrable in the sense of Definition \ref{def:superalge}.
\end{proof}

Finally, to conclude this subsection, consider $A$ to be an Abelian reductive connected subgroup of $G$. Under this assumption, with the intermediate Poisson algebra $S(\mathfrak{g})^A$, we will consider the superintegrable system chain. 

\begin{theorem}
\label{thm:abelianredsuper}
Let $G$ be a connected complex semisimple Lie group with Lie algebra $\mathfrak{g}$. Let $A\subset G$ be a 
connected Abelian reductive subgroup, and let its Lie algebra be $\mathfrak{a}  = \mathrm{Lie}(A)$ with $s:=\dim \mathfrak{a}$.  
For every $H\in\mathfrak{a}$ and $X \in \mathfrak{g}$, define the linear polynomial $\mu_H\in S(\mathfrak{g})$ by $\mu_H(X) := B(X,H)$. Let \begin{align*}
\mu_A:\mathfrak{g}\longrightarrow \mathfrak{a}^*
\end{align*} be the linear map $\mu_A(X)(H):= B(X,H)$. For a basis $H_1,\dots,H_s$ of $\mathfrak{a}$, define \begin{align*}
\mathcal{B}_A :=   \mu_A^*\mathbb{C}[\mathfrak{a}^*] = \mathbb{C}[\mu_{H_1},\dots,\mu_{H_s}] \subset S(\mathfrak{g}).
\end{align*} Then $\mathcal{B}_A\subset Z \left(S(\mathfrak{g})^A\right)$, where $Z(S(\mathfrak{g})^A)$ denotes the Poisson center of the Poisson algebra $S(\mathfrak{g})^A$. 
Consequently, the inclusion chain of Poisson algebras \begin{align*}
\mathcal{B}_A\subset S(\mathfrak{g})^A\subset S(\mathfrak{g})
\end{align*}
defines a superintegrable system in the sense of Definition \ref{def:superal}. Equivalently, the induced chain of affine Poisson varieties \begin{align*}
\mathfrak{g} \xrightarrow{\ \chi_A\ } \mathfrak{g}//A \xrightarrow{\ \overline{\mu}_A\ } \mathfrak{a}^* \cong \mathfrak{a}
\end{align*} is superintegrable, where $\chi_A:\mathfrak{g}\longrightarrow \mathfrak{g}//A:=\mathrm{Spec} \, S(\mathfrak{g})^A$ is the affine quotient morphism, and $\overline{\mu}_A$ is the unique morphism satisfying $\mu_A = \overline{\mu}_A \circ \chi_A$. In coordinates, for a basis $H_1,\dots,H_s$ of $\mathfrak{a}$, we have \begin{align*}
\overline{\mu}_A([X]_A) = \bigl( B(X,H_1),\dots,B(X,H_s)
\bigr).
\end{align*}
\end{theorem}

\begin{proof}
Since $A$ is reductive, the invariant algebra $S(\mathfrak{g})^A$ is finitely generated. Hence, $\mathfrak{g}//A:=\mathrm{Spec}S(\mathfrak{g})^A$ is an affine variety, and the inclusion $S(\mathfrak{g})^A\hookrightarrow S(\mathfrak{g})$ induces the quotient morphism $\chi_A:\mathfrak{g}\longrightarrow \mathfrak{g}//A$. 

We first prove that $\mu_H$ is $A$-invariant for every $H\in\mathfrak{a}$. Let $a \in A$, $H\in\mathfrak{a}$, and $X \in\mathfrak{g}$. Since $A$ is Abelian, the adjoint action of $A$ on its Lie algebra $\mathfrak{a}$ is trivial, and therefore $\mathrm{Ad}(a) H= H$. Using the $\mathrm{Ad}(G)$-invariance of $B$, we obtain \begin{align*}
\mu_H(\mathrm{Ad}(a)X) = B(\mathrm{Ad}(a)X,H) = B(X,\mathrm{Ad}(a^{-1})H) = B(X,H) = \mu_H(X).
\end{align*} Thus, $\mu_H \in S(\mathfrak{g})^A$ for every $H\in\mathfrak{a}$, and hence $\mathcal{B}_A\subset S(\mathfrak{g})^A$. 

We next prove that $\mathcal{B}_A$ lies in the Poisson center of $S(\mathfrak{g})^A$. For $f\in S(\mathfrak{g})$ and $V \in \mathfrak{g}$, let $\nabla f(X)\in\mathfrak{g}$ be defined by  $df_X(V)=B(\nabla f(X),V)$. Since $\nabla\mu_H=H $, under the identification $\mathfrak{g}\simeq\mathfrak{g}^*$, by the definition of the Lie-Poisson bracket, we have \begin{align}
\{\mu_H,f\}(X) = B \left(X,[H,\nabla f(X)]\right). \label{eq:poissonbracettrivia}
\end{align} By the $\mathrm{Ad}(G)$-invariance and symmetry of $B$, \begin{align*}
B \left(X,[H,\nabla f(X)]\right) = B([X,H],\nabla f(X)) = -B(\nabla f(X),[H,X]) = -df_X([H,X]).
\end{align*}
Now assume $f\in S(\mathfrak{g})^A$. Then $f$ is $A$-invariant, hence also invariant under the identity component $A^\circ$ of $A$. Since $H\in\mathfrak{a}=\mathrm{Lie}(A^\circ)$, the one-parameter subgroup $\exp(tH)\subset A^\circ$ acts trivially on $f$. 
Hence, \begin{align*}
f(\mathrm{Ad}(\exp(tH))X)=f(X)
\end{align*} for all $t$, and differentiating at $t=0$ gives  $df_X([H,X])=0$. Therefore, back to \eqref{eq:poissonbracettrivia}, we have  $\{\mu_H,f\}(X) = 0$  for all $X\in\mathfrak{g}$. Since $f\in S(\mathfrak{g})^A$ was arbitrary, we obtain $\{\mu_H,f\}=0$  for all $f\in S(\mathfrak{g})^A$. Thus \begin{align*}
\mu_H\in Z \left(S(\mathfrak{g})^A\right)
\end{align*} for every $H\in\mathfrak{a}$. Since $\mathcal{B}_A$ is generated by the functions $\mu_{H_1},\dots,\mu_{H_s}$, it follows that $\mathcal{B}_A\subset Z \left(S(\mathfrak{g})^A\right)$. 

We now compute the transcendence degrees. The map $\mu_A:\mathfrak{g}\longrightarrow \mathfrak{a}^*$ is surjective. Indeed, the dual map \begin{align*}
\mu_A^*:\mathfrak{a}\longrightarrow \mathfrak{g}^*, \qquad H\longmapsto B(\cdot,H),
\end{align*} is injective because $B$ is nondegenerate. Therefore the linear functions \begin{align*}
\mu_{H_1},\dots,\mu_{H_s}
\end{align*} are algebraically independent, and hence $\mathcal{B}_A\simeq \mathbb{C}[\mathfrak{a}^*]$ is a polynomial algebra in $s$ variables. Consequently, $\mathrm{trdeg}\, \mathcal{B}_A = s$. 

It remains to compute $\mathrm{trdeg \,}S(\mathfrak{g})^A$. Let $A^\circ$ be the identity component of $A$ and set $\Gamma = A/A^\circ$. Then $\Gamma$ is finite. Let $R:= S(\mathfrak{g})^{A^\circ}$. The quotient group $\Gamma$ acts on $R$ by algebra automorphisms and we have \begin{align*}
S(\mathfrak{g})^A = \left(S(\mathfrak{g})^{A^\circ}\right)^{A/A^\circ}=R^{\Gamma}.
\end{align*} We claim that the inclusion $R^{\Gamma}\subset R$ is integral (See, for instance, \cite[Section 5, Proposition 5.1]{AtiyahMacdonald}). Let $f\in R$. Define \begin{align*}
    P_f(T):=\prod_{\gamma\in\Gamma}\bigl(T-\gamma\cdot f\bigr)\in R[T].
\end{align*} Then $P_f(T)$ is monic and clearly $P_f(f)=0$. We now show that $P_f(T)\in R^{\Gamma}[T]$. Indeed, for any $\delta\in\Gamma$ we have \begin{align*}
    \delta\cdot P_f(T)=\prod_{\gamma\in\Gamma}\bigl(T-\delta\cdot(\gamma\cdot f)\bigr)=\prod_{\gamma\in\Gamma}\bigl(T-(\delta\gamma)\cdot f\bigr)=P_f(T),
\end{align*} since $\gamma \mapsto \delta\gamma$ is a permutation of the finite set $\Gamma$. Hence all coefficients of $P_f(T)$ are $\Gamma$-invariant, i.e., $P_f(T)\in R^{\Gamma}[T] = S(\mathfrak{g})^A[T]$. Therefore $f$ is integral over $R^{\Gamma}$, and consequently the extension \begin{align*}
S(\mathfrak{g})^A = \left(S(\mathfrak{g})^{A^\circ}\right)^{A/A^\circ} \subset S(\mathfrak{g})^{A^\circ}
\end{align*} is integral. Hence, taking invariants under the finite group $A/A^\circ$ does not change the transcendence degree, and therefore \begin{align*}
\mathrm{trdeg}S(\mathfrak{g})^A = \mathrm{trdeg}S(\mathfrak{g})^{A^\circ}.
\end{align*} 

To find the transcendence degree of $S(\mathfrak{g})^{A^\circ}$, by the dimension formula of an affine quotient, we now compute the dimension of $A^\circ$-orbit, which is $\dim A^\circ - \dim {A^\circ}_X$. In other words, we now show the stabilizer ${A^\circ}_X$ is finite. Since the torus $A^\circ$ acts on $\mathfrak{g}$, so there is a weight-space decomposition  $\mathfrak{g} = \bigoplus_{\lambda\in\Lambda}\mathfrak{g}_\lambda$, where $\Lambda\subset\mathfrak{a}^*$ is the finite set of weights. The infinitesimal kernel of the $A^\circ$-action on $\mathfrak{g}$ is \begin{align*}
\bigcap_{\lambda \in \Lambda}\ker\lambda.
\end{align*} If $H\in\mathfrak{a}$ lies in this intersection, then $[H,Y] =0$ for all $Y\in\mathfrak{g}$. Hence, $H\in Z(\mathfrak{g})$. Since $\mathfrak{g}$ is semisimple, $Z(\mathfrak{g})=0$. Therefore, the infinitesimal $A^\circ$-action on $\mathfrak{g}$ is faithful.

Choose $X\in\mathfrak{g}$ such that  \begin{align}
X =  \sum_{\lambda\in\Lambda}X_\lambda,\qquad X_\lambda\in\mathfrak{g}_\lambda, \label{eq:summationX}
\end{align} we have $X_\lambda\neq 0$ for every nonzero weight $\lambda\in\Lambda\setminus\{0\}$. Such elements form a nonempty Zariski open subset of $\mathfrak{g}$. If $H\in\mathfrak{a}$ stabilizes such an $X$, then $[H,X] = 0$. Using the expression \eqref{eq:summationX}, we  find $\lambda(H)X_\lambda = 0$ for every nonzero weight $\lambda$. Since $X_\lambda \neq 0$, it follows that $\lambda(H) = 0$. Since the nonzero weights span $\mathfrak{a}^*$, we conclude that $H=0$. Thus, the generic stabilizer of the $A^\circ$-action is finite, and the generic $A$-orbit has dimension $d_A = s$. By the dimension formula for affine quotients \eqref{eq:gitquotient}, \begin{align*}
\mathrm{trdeg} \, S(\mathfrak{g})^A = \dim\mathfrak{g}-d_A = \dim\mathfrak{g}-s.
\end{align*}

Combining the two transcendence-degree identities gives \begin{align*}
\mathrm{trdeg} \, S(\mathfrak{g})^A + \mathrm{trdeg} \ \mathcal{B} _A = (\dim \mathfrak{g}-s) + s = \dim\mathfrak{g} = \mathrm{trdeg} \, S(\mathfrak{g}).
\end{align*} Together with $  \mathcal{B}_A\subset Z \left(S(\mathfrak{g})^A\right)$, this proves that the inclusion relations \begin{align*}
\mathcal{B}_A\subset S(\mathfrak{g})^A\subset S(\mathfrak{g})
\end{align*} is a superintegrable inclusion chain in the sense of Definition \ref{def:superalge}. Equivalently, the induced chain \begin{align*}
\mathfrak{g} \xrightarrow{\ \chi_A\ } \mathfrak{g}//A
\xrightarrow{\ \overline{\mu}_A\ } \mathfrak{a}^*
\end{align*} is a superintegrable chain of affine Poisson varieties. The proof is complete.
\end{proof}

\begin{remark} 
\label{rem:abvsmaximaltorus}
When $A=T$ is a maximal torus, Theorem \ref{thm:abelianredsuper} gives the central base algebra \begin{align*}
B_T = \mathbb{C}[\mu_{H_1},\dots,\mu_{H_r}] \simeq \mathbb{C}[\mathfrak{t}^*], \qquad r=\dim T=\mathrm{rank}\mathfrak{g}.
\end{align*}  Thus, we obtain a superintegrable chain  $\mathfrak{g} \longrightarrow \mathfrak{g}//T \longrightarrow \mathfrak{t}^*$. Moreover, in this case, we may also consider the chain from Theorem \ref{thm:superalchain}, \begin{align*}
\mathfrak{g} \longrightarrow \mathfrak{g}//T \longrightarrow \mathfrak{g}//G,
\end{align*} whose base algebra is $S(\mathfrak{g})^G$. This alternative is still compatible with the dimension characteristic as \begin{align*}
\mathrm{trdeg} \, S(\mathfrak{g})^G = \mathrm{rank}\, \mathfrak{g} = \dim T.
\end{align*}

For a proper Abelian reductive subgroup $A\subsetneq T$, however, the chain  $\mathfrak{g} \longrightarrow \mathfrak{g}//A \longrightarrow \mathfrak{g}//G$ is generally not superintegrable. Indeed,  $\mathrm{trdeg} \, S(\mathfrak{g})^G = \mathrm{rank}\, \mathfrak{g}$. By the dimension characteristic, with the intermediate algebra $S(\mathfrak{g})^A$, the base algebra needs to have a transcendence degree of $d_A = \dim A= \dim\mathfrak{a}$. Thus, for a general Abelian reductive subgroup $A$, the natural base algebra is not usually $S(\mathfrak{g})^G$. Instead, we may take \begin{align*}
\mathcal{B}_A = \mu_A^*\mathbb{C}[\mathfrak{a}^*] = \mathbb{C}[\mu_{H_1},\dots,\mu_{H_s}] \subset Z \left(S(\mathfrak{g})^A\right),
\end{align*}  which has the required transcendence degree $\mathrm{trdeg} \ \mathcal{B} _A = s = d_A$. 
\end{remark}

\subsection{Mishchenko-Fomenko subalgebras inside \texorpdfstring{$S(\mathfrak{g})^A$}{S(\mathfrak{g})A}}
\label{subsec:mfinsideA}

Throughout this Subsection \ref{subsec:mfinsideA}, let $G$ be connected and semisimple over $\mathbb{C}$, let $A \subset G$ be a closed reductive subgroup, and let $B (\cdot,\cdot)$ be the fixed $\mathrm{Ad}(G)$-invariant nondegenerate symmetric bilinear form used to identify $\mathfrak{g}$ with $\mathfrak{g}^*$. In this subsection, we will consider the relations between the  Mischenko-Fomenko algebra and the Poisson center of the centraliser $S(\mathfrak{g})^A$. We would like to know under what condition the Mischenko-Fomenko algebra can be used as the commutative subalgebra satisfying the superintegrable Poisson projection chain. We define the Lie algebra centraliser by  \begin{align*}
\mathfrak{g}^A := \{X \in \mathfrak{g} : \mathrm{Ad}(a)X = X \text{ for all } a \in A\},
\end{align*}
We denote by $\mathfrak{g}_{\mathrm{reg}}$ the set of regular elements of $\mathfrak{g}$.  Here \begin{align*}
    \mathfrak{g}_{\mathrm{reg}} := \left\{X \in \mathfrak{g} : \dim \mathfrak{g}_X=\mathrm{rank} \,\mathfrak{g}\right\}, \quad \text{with } \mathfrak{g}_X = \left\{y\in \mathfrak{g}: [x,y] = 0\right\}
\end{align*}Moreover, we write $Z_P(S(\mathfrak{g})^A)$ for the Poisson center of $S(\mathfrak{g})^A$. Define \begin{align*}
\mathcal{C}_A = S(\mathfrak{g})^A, \qquad d_A = \dim \mathfrak{g} - \mathrm{trdeg} \,\mathcal{C}_A.
\end{align*}
Recall that a base algebra for the superintegrable chain through $\mathcal{C}_A$ is a finitely generated subalgebra $\mathcal{B} \subset Z_P(\mathcal{C}_A)$ such that
\begin{align*}
\mathrm{trdeg} \, \mathcal{B} = d_A.
\end{align*}

We now recall a basic fact about the Mishchenko-Fomenko subalgebra defined in \cite{MF78}. For a regular shift $\mu \in \mathfrak{g}^*$, the \textit{argument shift} construction produces a Poisson-commutative algebra of maximal possible size, and in fact, it is freely generated by the expected number of elements. Let $P_1,\dots,P_\ell$ be homogeneous, algebraically independent generators of $S(\mathfrak{g})^G$, where $\ell = \mathrm{rank} \, \mathfrak{g}$. For $\mu \in \mathfrak{g}^*$ and $f \in S(\mathfrak{g})$, define \begin{align*}
D_\mu^j(f)(X) := \left.\dfrac{d^j}{dt^j}\right\vert_{t= 0} f(X + t\mu).
\end{align*} 
The corresponding Mishchenko-Fomenko subalgebra is
\begin{align*}
\mathcal{F}_\mu := \mathbb{C}\big[D_\mu^j(P_i) : 1 \leq i \leq \ell,\ 0 \leq j \leq \deg P_i - 1\big].
\end{align*}
Equivalently, $\mathcal{F}_\mu$ is generated by the coefficients of all polynomials $P(X+t\mu)$ with $P \in S(\mathfrak{g})^G$. 

\begin{theorem}
\label{thm:mfalgebra} \cite{MF78}
If $\mu \in \mathfrak{g}_{\mathrm{reg}}$, then $\mathcal{F}_\mu$ is a polynomial Poisson-commutative algebra on \begin{align*}
b(\mathfrak{g}) := \frac{1}{2}\big(\dim \mathfrak{g} + \mathrm{rank}\mathfrak{g}\big)
\end{align*} generators. Equivalently, we have \begin{align*}
\mathcal{F}_\mu \cong \mathbb{C}[t_1,\dots,t_{b(\mathfrak{g})}].
\end{align*}
\end{theorem}

\begin{proposition}
\label{prop:mfcriterion}
For $\mu \in \mathfrak{g}$, $\mu \in \mathfrak{g}^A$ if and only $\mathcal{F}_\mu \subset S(\mathfrak{g})^A$. 
\end{proposition}

\begin{proof}
Assume first that $\mu \in \mathfrak{g}^A$. Let $P \in S(\mathfrak{g})^G$. For every $a \in A$ and every $X \in \mathfrak{g}$, we have \begin{align*}
P(\mathrm{Ad}(a)X+t\mu) = P(\mathrm{Ad}(a)X+t\mathrm{Ad}(a)\mu) = P(\mathrm{Ad}(a)(X+t\mu)) = P(X+t\mu).
\end{align*} Hence, every coefficient of the polynomial $P(X+t\mu)$ is $A$-invariant. Since $\mathcal{F}_\mu$ is generated by all such coefficients, it follows that  $\mathcal{F}_\mu \subset S(\mathfrak{g})^A$. 

Assume conversely that $\mathcal{F}_\mu \subset S(\mathfrak{g})^A$. Consider the quadratic invariant \begin{align*}
Q(X) := \dfrac{1}{2}B (X,X) 
\end{align*} defined in \cite{Vin91}.
Since $\kappa$ is $\mathrm{Ad}(G)$-invariant, we find $Q \in S(\mathfrak{g})^G$. The coefficient of $t$ in $Q(X+t\mu)$ is \begin{align*}
D_\mu^1(Q)(X) = B (X,\mu).
\end{align*} By definition of $\mathcal{F}_\mu$, this function belongs to $\mathcal{F}_\mu$, hence to $S(\mathfrak{g})^A$. Therefore, for every $a \in A$ and every $X \in \mathfrak{g}$, we deduce $B (\mathrm{Ad}(a)X,\mu) = B (X,\mu)$. Using $\mathrm{Ad}(G)$-invariance of $\kappa$, this is equivalent to \begin{align*}
B (X,\mathrm{Ad}(a^{-1})\mu) = B (X,\mu).
\end{align*}
Since $B$ is nondegenerate, we obtain $\mathrm{Ad}(a^{-1})\mu = \mu$. Hence, $\mu \in \mathfrak{g}^A$.
\end{proof}

\begin{corollary}
\label{cor:mfregexi}
There exists a regular element $\mu$ such that $\mathcal{F}_\mu \subset S(\mathfrak{g})^A$ if and only if $\mathfrak{g}^A \cap \mathfrak{g}_{\mathrm{reg}} \ne \varnothing$. In particular, if $A=T$ is a maximal torus, then every $\mu \in \mathfrak{t} \cap \mathfrak{g}_{\mathrm{reg}}$ satisfies $\mathcal{F}_\mu \subset S(\mathfrak{g})^T$. 
\end{corollary}

\begin{proof}
This is immediate from Proposition \ref{prop:mfcriterion}. If $A=T$, then $\mathfrak{g}^T = \mathfrak{t}$.
\end{proof}

\begin{corollary}
    \label{coro:mfnoconb}
Let $\mu \in \mathfrak{g}$ be a regular semisimple. Then, no base algebra $\mathcal{B} \subset Z_P(\mathcal{C}_A)$ such that $\mathcal{B} = \mathcal{F}_\mu$. 
\end{corollary} 

\begin{proof}
Assume, for contradiction, that $\mathcal{B} = \mathcal{F}_\mu$. Since $ \mathcal{B}$ is a commutative base algebra, we deduce  $\mathrm{trdeg} \, \mathcal{B} = d_A$. Since $\mu$ is regular semisimple, the Mishchenko-Fomenko algebra has maximal possible transcendence degree, hence \begin{align*}
\mathrm{trdeg}\, \mathcal{F}_\mu = \frac{1}{2}(\dim \mathfrak{g} + \mathrm{rank} \,\mathfrak{g}).
\end{align*} Therefore, \begin{align}
d_A = \frac{1}{2}(\dim \mathfrak{g} + \mathrm{rank} \,\mathfrak{g}). \label{eq:da}
\end{align} By the definition of $d_A$, this implies \begin{align}
\mathrm{trdeg}\, \mathcal{C}_A = \dim \mathfrak{g} - d_A = \frac{1}{2}(\dim \mathfrak{g} - \mathrm{rank}\mathfrak{g}). \label{eq:rankaa}
\end{align} On the other hand, $\mathcal{B} \subset Z_P(\mathcal{C}_A)$ implies $\mathcal{B} \subset \mathcal{C}_A$, and therefore \begin{align}
\mathrm{trdeg} \, \mathcal{F}_\mu = \mathrm{trdeg} \ \mathcal{B}  \leq \mathrm{trdeg} \, \mathcal{C}_A.\label{eq:rankmf}
\end{align} Substituting the identities \eqref{eq:da} and \eqref{eq:rankaa} into \eqref{eq:rankmf} yields \begin{align*}
\frac{1}{2}(\dim \mathfrak{g} + \mathrm{rank}\mathfrak{g}) \leq \frac{1}{2}(\dim \mathfrak{g} - \mathrm{rank}\mathfrak{g}),
\end{align*} which is impossible for nonzero semisimple $\mathfrak{g}$.
\end{proof}

\begin{remark}
\label{rmk:mfnotbase}
A regular semisimple Mishchenko-Fomenko subalgebra can never be a base algebra for the superintegrable chain through $S(\mathfrak{g})^A$.
\end{remark}

\begin{theorem}
\label{thm:mfinsidesuper}
Assume that $\mathcal{B} \in \mathcal{B}_A$, and let $\mu \in \mathfrak{g}^A \cap \mathfrak{g}_{\mathrm{reg}}$. Then the following statements hold.

\begin{enumerate}
\item The chain
\begin{align*}
\mathfrak{g} \xrightarrow{\ \chi_A\ } \mathfrak{g}//A \xrightarrow{\ \rho_B\ } \mathrm{Spec}\, \mathcal{B}
\end{align*}
is superintegrable in the sense of Definition \ref{def:superalge}.

\item We have $\mathcal{F}_\mu \subset S(\mathfrak{g})^A$. 

\item The inclusion of Poisson algebras $\mathcal{F}_\mu \hookrightarrow S(\mathfrak{g})^A$ induces a Poisson morphism \begin{align*}
\psi_{\mu,A} : \mathfrak{g}//A \longrightarrow \mathrm{Spec}\,\mathcal{F}_\mu.
\end{align*}
\end{enumerate}

Therefore, whenever $\mathfrak{g}//A$ occurs as the intermediate Poisson variety of a superintegrable chain, every regular element
\begin{align*}
\mu \in \mathfrak{g}^A \cap \mathfrak{g}_{\mathrm{reg}}
\end{align*}
produces a distinguished Poisson-commutative subalgebra $\mathcal{F}_\mu$ inside the intermediate algebra $S(\mathfrak{g})^A$.
\end{theorem}

\begin{proof}
We first show part (1). Since $ \mathcal{B}\in\mathcal{B}_A$, Proposition \ref{prop:superinchainA} applies and yields that the sequence of morphisms
\begin{align*}
\mathfrak{g} \xrightarrow{\ \chi_A\ } \mathfrak{g}//A \xrightarrow{\ \rho_B\ } \mathrm{Spec} \, \mathcal{B}
\end{align*}
forms a superintegrable chain in the sense of Definition \ref{def:superalge}. 

We now show the second part. As $\mu\in\mathfrak{g}^A$, Proposition \ref{prop:mfcriterion} implies that $\mathcal{F}_\mu\subset S(\mathfrak{g})^A$.  

Finally, by Theorem \ref{thm:mfalgebra}, $\mathcal{F}_\mu$ is Poisson-commutative in $S(\mathfrak{g})$, i.e., the Poisson bracket on $S(\mathfrak{g})$ restricts to a Poisson bracket on $\mathcal{F}_\mu$.
By part (2),  $\mathcal{F}_\mu$ is a Poisson subalgebra of the Poisson algebra $S(\mathfrak{g})^A$. In particular, the inclusion \begin{align*}
\iota:\mathcal{F}_\mu \hookrightarrow S(\mathfrak{g})^A
\end{align*} preserves both the commutative algebra structure and the Poisson bracket. Hence, it is a morphism of Poisson algebras. Applying $\mathrm{Spec}(\cdot)$ to $\iota$ yields a morphism of affine varieties
\begin{align*}
\psi_{\mu,A}:\mathfrak{g}//A=\mathrm{Spec} S(\mathfrak{g})^A \longrightarrow \mathrm{Spec} \,\mathcal{F}_\mu.
\end{align*}
Moreover, by construction, the pullback $\psi_{\mu,A}^*:\mathcal{F}_\mu\to S(\mathfrak{g})^A$ equals $\iota$ and is Poisson. Therefore, $\psi_{\mu,A}$ is a Poisson morphism. This proves all the claims.
\end{proof}

\begin{remark}
\label{rem:mf-not-base}
The algebra $\mathcal{F}_\mu$ should not be identified with the base algebra $ \mathcal{B}$ from Definition \ref{def:superalge}. The base algebra is required to satisfy
\begin{align*}
\mathcal{B} \subset Z(S(\mathfrak{g})^A),
\end{align*}
whereas $\mathcal{F}_\mu$ is only a Poisson-commutative subalgebra of $S(\mathfrak{g})^A$. Thus, in the superintegrable chain, $ \mathcal{B}$ remains the base algebra, while $\mathcal{F}_\mu$ is an additional Poisson-commutative subalgebra inside the intermediate algebra.
\end{remark}

\begin{corollary}
\label{cor:mfspbase}
Let  $d_A = \mathrm{rank} \,\mathfrak{g}$, and let $\mu \in \mathfrak{g}^A \cap \mathfrak{g}_{\mathrm{reg}}$. Then \begin{align*}
S(\mathfrak{g})^G \subset \mathcal{F}_\mu \subset S(\mathfrak{g})^A \subset S(\mathfrak{g}),
\end{align*}
and the chain
\begin{align*}
\mathfrak{g} \xrightarrow{\ \chi_A\ } \mathfrak{g}//A \xrightarrow{\ \rho\ } \mathfrak{g}//G
\end{align*}
is superintegrable, where $\rho$ is induced by the inclusion $S(\mathfrak{g})^G \subset S(\mathfrak{g})^A$. 
\end{corollary}

\begin{proof}
We first verify that $S(\mathfrak{g})^G$ is a base algebra for the chain through $S(\mathfrak{g})^A$.
Since $A\subset G$, invariance under $G$ implies invariance under $A$, and hence
\begin{align*}
S(\mathfrak{g})^G\subset S(\mathfrak{g})^A.
\end{align*}
The Poisson bracket on $S(\mathfrak{g})$ restricts to a Poisson bracket on the subalgebra $S(\mathfrak{g})^A$. Moreover, $S(\mathfrak{g})^G$ is the Poisson center of $S(\mathfrak{g})$. Thus, for every $f\in S(\mathfrak{g})^G$ and every $h\in S(\mathfrak{g})^A$,  $\{f,h\}=0$. Equivalently, \begin{align*}
S(\mathfrak{g})^G \subset Z_P\!\left(S(\mathfrak{g})^A\right),
\end{align*}
where $Z_P(\cdot)$ denotes the Poisson center.

As $\mathfrak{g}$ is complex semisimple, Chevalley's theorem identifies $S(\mathfrak{g})^G$ as a polynomial algebra on $\ell = \mathrm{rank}\,\mathfrak{g}$ homogeneous generators. In particular, \begin{align*}
\mathrm{trdeg}\, S(\mathfrak{g})^G = \mathrm{rank}\,\mathfrak{g} = d_A.
\end{align*} Therefore, $ \mathcal{B}_0: =S(\mathfrak{g})^G$ is a base algebra in the sense of Definition \ref{def:superalge}. Let \begin{align*}
\rho:=\operatorname{Spec}(\iota):\mathfrak{g}//A=\operatorname{Spec}S(\mathfrak{g})^A \longrightarrow \operatorname{Spec}S(\mathfrak{g})^G=\mathfrak{g}//G
\end{align*} be the morphism induced by the inclusion of $\mathbb{C}$-algebras $\iota:S(\mathfrak{g})^G\hookrightarrow S(\mathfrak{g})^A$. Proposition \ref{prop:superinchainA} applied to $ \mathcal{B}_0$ then implies that \begin{align*}
\mathfrak{g} \xrightarrow{\ \chi_A\ } \mathfrak{g}//A \xrightarrow{\ \rho\ } \mathfrak{g}//G
\end{align*} is a superintegrable chain.

Next, we show the inclusion with $\mathcal{F}_\mu$. By definition, \begin{align*}
\mathcal{F}_\mu = \mathbb{C}\big[D_\mu^j(P_i) : 1\leq i\leq \ell,\ 0\leq j\leq \deg P_i-1\big].
\end{align*}  Taking $j=0$ shows $P_i=D_\mu^0(P_i)\in\mathcal{F}_\mu$ for all $i$, hence $S(\mathfrak{g})^G\subset\mathcal{F}_\mu$. Finally, since $\mu\in\mathfrak{g}^A$, Proposition \ref{prop:mfcriterion} yields $\mathcal{F}_\mu\subset S(\mathfrak{g})^A$.
\end{proof}

\begin{remark}
\label{rem:refinedsuperintegrable}
More generally, given a superintegrable chain $M\xrightarrow{\pi_1}\mathcal{P}\xrightarrow{\pi_2}\mathcal{B}$, we call it \emph{refined} if we choose an intermediate Poisson-commutative subalgebra \begin{align*}
    C^\infty(\mathcal{B})\subset \mathcal{F}\subset C^\infty(\mathcal{P}),
\end{align*}
such that the projection $\pi_2$ factors through the associated affine Poisson variety $\mathrm{Spec}\,\mathcal{F}$. In other words, $\pi_2$ is the composite \begin{align*}
     \mathcal{P}\longrightarrow \mathrm{Spec}\,\mathcal{F}\longrightarrow \mathcal{B}.
 \end{align*} See more detailed definition on the refined superintegrable  system in \cite{ArthamonovReshetikhin2021}.

In the situation of Corollary \ref{cor:mfspbase}, the Mishchenko-Fomenko subalgebra $\mathcal{F}_\mu$ provides a refinement of the base $S(\mathfrak{g})^G$. Hence, we obtain the \emph{refined superintegrable system} \begin{align*}
    \mathfrak{g}\xrightarrow{\chi_A}\mathfrak{g}//A \xrightarrow{\ \psi_{\mu,A}\ }\mathrm{Spec}\,\mathcal{F}_\mu \xrightarrow{\ \rho_{\mu}\ }\mathfrak{g}//G,
\end{align*} where $\psi_{\mu,A}$ is induced by the inclusion $\mathcal{F}_\mu\hookrightarrow S(\mathfrak{g})^A$ and $\rho_{\mu}$ is induced by $S(\mathfrak{g})^G\hookrightarrow \mathcal{F}_\mu$ such that $\rho = \rho_{\mu}\circ \psi_{\mu,A}$.
\end{remark}

\begin{corollary}
\label{cor:mfspectrum}
Assume that $\mu \in \mathfrak{g}_{\mathrm{reg}}$. Then we have $\mathrm{Spec}\,\mathcal{F}_\mu \cong \mathbb{A}^{b(\mathfrak{g})}$, where $b(\mathfrak{g}) = \frac{1}{2}\big(\dim \mathfrak{g} + \mathrm{rank}\mathfrak{g}\big)$. In particular, if $\mu \in \mathfrak{g}^A \cap \mathfrak{g}_{\mathrm{reg}}$, then $\psi_{\mu,A} : \mathfrak{g}//A \longrightarrow \mathbb{A}^{b(\mathfrak{g})}$ is a Poisson morphism.
\end{corollary}

\begin{proof}
Assume $\mu\in\mathfrak{g}_{\mathrm{reg}}$. By Theorem \ref{thm:mfalgebra}, the Mishchenko-Fomenko algebra $\mathcal{F}_\mu$ is a polynomial algebra on $b(\mathfrak{g})$ algebraically independent generators. That is, there exists a $\mathbb{C}$-algebra isomorphism $\mathcal{F}_\mu \cong \mathbb{C}[t_1,\dots,t_{b(\mathfrak{g})}]$. Applying the contravariant functor $\operatorname{Spec}$ yields an isomorphism of affine varieties
\begin{align*}
\operatorname{Spec}\mathcal{F}_\mu \cong \operatorname{Spec}\mathbb{C}[t_1,\dots,t_{b(\mathfrak{g})}] = \mathbb{A}^{b(\mathfrak{g})}.
\end{align*}

For the final statement, assume in addition that $\mu\in\mathfrak{g}^A$. Theorem \ref{thm:mfinsidesuper} constructs a Poisson algebra morphism $\mathcal{F}_\mu\hookrightarrow S(\mathfrak{g})^A$; hence, by taking spectra, a Poisson morphism
\begin{align*}
\psi_{\mu,A}:\mathfrak{g}//A=\operatorname{Spec}S(\mathfrak{g})^A \longrightarrow \operatorname{Spec}\mathcal{F}_\mu.
\end{align*}
Composing $\psi_{\mu,A}$ with the identification $\operatorname{Spec}\mathcal{F}_\mu\cong \mathbb{A}^{b(\mathfrak{g})}$, $\psi_{\mu,A}:\mathfrak{g}//A\to\mathbb{A}^{b(\mathfrak{g})}$ is a Poisson morphism.
\end{proof}

\subsection{Symplectic leaves in \texorpdfstring{$\mathfrak{g}//T$}{\mathfrak{g}//T}}
\label{sec:leaves}
 
Let $\left(\mathfrak{g}//T\right)_{\mathrm{reg}}$ be the regular subset of $\mathfrak{g}//T$. In this Section \ref{sec:leaves}, we will examine the symplectic leaves in $\left(\mathfrak{g}//T\right)_{\mathrm{reg}}$ by studying the pull back of the regular element in $\mathfrak{g}//G$ via $\rho^{-1}$.   Let $G$ be a semisimple and connected Lie group with a Lie algebra $\mathfrak{g}$. Fix an $\mathrm{Ad}(G)$-invariant inner product $ B( \cdot , \cdot )$ on $\mathfrak{g}$. For any $H\in\mathfrak{t}$ and $u\in\mathfrak{g}$, define the linear function:
\begin{align*}
\mu_H(u): = B(u,H)  .
\end{align*} Note that $\mu_H$ can be considered linear generators for the Cartan commutant $S(\mathfrak{g})^T$. A similar argument can be found in \cite[Proposition 2.5]{campoamor2026construction}.
\begin{lemma} 
\label{lem:mucenter}
For every $H\in\mathfrak{t}$, we have $\mu_H\in S(\mathfrak{g})^T$ and $\{\mu_H,f\} = 0$ for all $f\in S(\mathfrak{g})^T$.
\end{lemma}
\begin{proof}
See the first part of the proof Theorem \ref{thm:abelianredsuper}. 
\end{proof}
  
Let $\{C_1,\dots,C_r\}\subset S(\mathfrak{g})^G$ be a set of homogeneous algebraically independent generators. By Theorem \ref{thm:superalchain}, \begin{align*}
\chi_T: \mathfrak{g}\longrightarrow \mathfrak{g}//T,\qquad \rho :  \mathfrak{g}//T\longrightarrow \mathfrak{g}//G,\qquad \chi:= \rho\circ \chi_T: \mathfrak{g}\longrightarrow \mathfrak{g}//G,
\end{align*} such that $\rho$ is the morphism induced by the inclusion $S(\mathfrak{g})^G\subset S(\mathfrak{g})^T$. Then, for any $c\in(\mathfrak{g}//G)_{\mathrm{reg}}$ and the corresponding coadjoint orbit $\mathcal{O}_c\subset\mathfrak{g}$, we have \begin{align*}
\rho^{-1}(c) = \chi_T(\mathcal{O}_c) \subset \mathfrak{g}//T, \qquad \chi_T^{-1} \big(\rho^{-1}(c)\big) = \chi^{-1}(c) = \{u\in\mathfrak{g}: C_i(u)=c_i\}.
\end{align*} Thus, the fiber of $\rho$ is the image of a single coadjoint orbit, and its pullback to $\mathfrak{g}$ is the usual Kostant-Chevalley level set.

By Remark \ref{rkm:ijc}, we see that $\rho^{-1}(c)$ is usually a union of symplectic leaves. To isolate a single leaf inside $\rho^{-1}(c)$, we must fix the $T$ moment map component. 
Consider the $T$-equivariant linear map \begin{align}
 \mu:\mathfrak{g}\longrightarrow \mathfrak{t}^* \cong \mathfrak{t},\qquad \mu(u)=u \vert_{\mathfrak{t}}. \label{eq:mu1}
\end{align} Since $\mu$ is $T$-invariant, it induces a regular map \begin{align}
 \bar{\mu}:\mathfrak{g}//T \longrightarrow \mathfrak{t},\qquad \bar{\mu}([u]_T)=u \vert_{\mathfrak{t}}. \label{eq:mu2}
\end{align}   We then define   \begin{align*}
  J:=(\rho,\bar{\mu}): \mathfrak{g}//T \longrightarrow \left(\mathfrak{g}// G\right) \times \mathfrak{t}  
\end{align*} by joining $\rho$ and $\bar{\mu}$ in \eqref{eq:mu2} such that each component in the fiber of $\rho$ are controlled.

\begin{proposition}
\label{prop:Jcenter}
Each component of $J$ is a Casimir in the center of $\mathfrak{g}//T$. Equivalently, the pullbacks $\rho^*(C_i)$ and the coordinate functions of $\bar{\mu}$ all lie in $\mathcal{Z} \bigl(S(\mathfrak{g})^T\bigr)$.
\end{proposition}
\begin{proof}
 The components of $\bar{\mu}$ are $\mu_H$ with $H\in\mathfrak{t}$, which are central in $S(\mathfrak{g})^T$ by Lemma \ref{lem:mucenter}.
\end{proof}

By Proposition \ref{prop:Jcenter}, each component of $J:\mathfrak{g}//T \rightarrow (\mathfrak{g}//G) \times \mathfrak{t}$ is a Casimir. Hence, all Hamiltonian vector fields on $\mathfrak{g}//T$ are tangent to the fibers of $J$. Therefore, the characteristic distribution of the Poisson structure is contained in $\ker dJ$, and every fiber $J^{-1}(c, \alpha)$ is a union of symplectic leaves.  
 On  $(\mathfrak{g}//T)_{\mathrm{reg}}$, the rank of $dJ$ is constant and equal to $2r$, and it is sufficient to show that for every regular pair $(c,\alpha)$, there exists a connected symplectic submanifold whose image is equal to $J^{-1}(c,\alpha)$. Then each connected component of that fiber is a symplectic leaf, and the Poisson rank is constant. At any regular point, the symplectic leaf has dimension \begin{align*}
  \dim(\mathfrak{g}//T)_{\mathrm{reg}}-\mathrm{rank}(dJ).
\end{align*}  Thus, each connected component of a regular fiber is a single symplectic leaf. 

We now provide a formal construction of the symplectic leaf in $\mathfrak{g}//T$. Let $\chi:\mathfrak{g}\to \mathfrak{g}// G$ be the Chevalley map given by $\chi(u)=(C_1(u),\dots,C_r(u))$. For a point $c=(c_1,\dots,c_r)\in \mathfrak{g}// G$, define
\begin{align*}
 \mathcal{O}_c:=\chi^{-1}(c),
\end{align*} which is a union of $\mathrm{Ad}(G)$-orbits. In particular, on the regular locus $\left( \mathfrak{g}// G\right)_{\mathrm{reg}}$ and $\mathcal{O}_c$, there is a single coadjoint orbit.  For each $\alpha$, we define the restricted moment map $$ \mu_T  := \left.\mu \right\vert_{\mathcal{O}_c}: \mathcal{O}_c \hookrightarrow \mathfrak{g}  \xrightarrow{\ \iota\ } \mathfrak{t}^* \cong \mathfrak{t}$$ such that $\iota_{X_H} \omega_{\mathrm{KKS},c} = d B(\mu_T,H)$ with $H \in \mathfrak{t}$,  given by  $\mu_T^c(\eta)=\eta\vert_{\mathfrak{t}} $. For $\alpha\in\mathfrak{t}^*$, define
\begin{align*}
 J^{-1}(c, \alpha)=\bigl\{[u]_T\in \mathfrak{g}//T:  C_i(u)=c_i, \text{ with } i=1,\dots,r,\ u \vert_{\mathfrak{t}}=\alpha\bigr\}.
\end{align*} In the rest of this subsection, we will work on the regular locus $(\mathfrak{g}//T)_{\mathrm{reg}}$, where we assume that the differentials $\{\,d(\rho^*C_i),\,d(\bar{\mu}_1),\dots,d(\bar{\mu}_r)\,\}$ have maximal rank $2r$, and the $T$-action on $\mu_{T,c}^{-1}(\alpha) := \mu^{-1}(\alpha)\cap\mathcal{O}_c$ is globally free on the regular locus. Then in the following theorem, we illustrate the symplectic leaf in $\mathfrak{g}//T$. 
Then, on the regular locus, \begin{align*}
\rho^{-1}(c) = \bigsqcup_{\alpha\in \mu_T(\mathcal{O}_c)} J^{-1}(c,\alpha).
\end{align*} In the following theorem, we provide the identification of the symplectic leaf $J^{-1}(c,\alpha)$ by identifying these fibers with reduced adjoint orbits. For the Marsden-Weinstein reduction on a submanifold, see, for instance, \cite{MR4497405}.

\begin{theorem}
\label{thm:leaves}
Let $c \in \left(\mathfrak{g}//T\right)_{\mathrm{reg}}$, and let $\alpha \in \mathfrak{t}^*$ is the regular value of $\mu_T:\mathcal{O}_c \rightarrow \mathfrak{t}^*$. Let $ \bigl(J^{-1}(c, \alpha)\bigr)_{\varpi}$ be the smooth and connected component of $J^{-1}(c,\alpha)$. On $(\mathfrak{g}//T)_{\mathrm{reg}}$, the symplectic leaves are precisely the connected components of the common level sets of $J$:
\begin{align*}
  \mathcal{L}_{c,\alpha} = \bigl(J^{-1}(c, \alpha)\bigr)_{\varpi} = \bigl\{[u]_T\in \mathfrak{g}//T:\ C_i(u)=c_i,\ u\vert_{\mathfrak{t}}=\alpha\bigr\}_{\varpi} .
\end{align*}
Equivalently, for any $c \in (\mathfrak{g}//G)_{\mathrm{reg}}$, suppose that $\mu_T:\mathcal{O}_c\to\mathfrak{t}^*$ is the $T$-moment map on any coadjoint orbit $(\mathcal{O}_c,\omega_{\mathrm{KKS},c})$. Let \begin{align}
    Z_{c,\alpha} := \mu_{T,c}^{-1}(\alpha) = \mathcal{O}_c \cap \mu^{-1}(\alpha).
\end{align} Suppose that $T$-action on $Z_{c,\alpha}$ is free.  The Marsden-Weinstein reduction is then given by
\begin{align*}
 \mathcal{L}_{c,\alpha}   \cong \bigl(\mathcal{O}_c \cap \mu^{-1}(\alpha)\bigr)\big/  T.  
\end{align*}
\end{theorem}
\begin{proof}
  Fix $c\in(\mathfrak{g}//G)_{\mathrm{reg}}$ and let $(\mathcal{O}_c,\omega_{\mathrm{KKS},c})$ be the coadjoint orbit. For $H \in \mathfrak{t}$, the fundamental vector field generated by $H$ at $\eta\in\mathcal{O}_c$ is $X_H(\eta)=[H,\eta]$. For any $X,Y \in \mathfrak{g}$, the KKS form at $c$ is given by \begin{align*}
  \omega_{\mathrm{KKS},c}(\eta)\big([X,\eta],[Y,\eta]\big)=B(\eta,[X,Y]).
\end{align*} We proof this theorem by showing the following parts: we first show that $\mathcal{O}_c//_\alpha T  :=\bigl(\mathcal{O}_c \cap \mu^{-1}(\alpha)\bigr)\big/  T$ is symplectic by providing its symplectic structure. We then show that $\mathcal{O}_c//_\alpha T$ to $\mathfrak{g}//T$ is a Poisson immersion. Finally, we explain why $\mathcal{O}_c//_\alpha T $ can be identified by the level set of $J$.

Recall that the $T$-action on $\mathcal{O}_c \subset \mathfrak{g}^*\cong\mathfrak{g}$ is Hamiltonian, with the moment map $\mu_T:\mathcal{O}_c \rightarrow \mathfrak{t}^*$ given by restriction to $\mathfrak{t}$. For a fixed regular value $\alpha \in \mathfrak{t}^*\cong\mathfrak{t}$ such that the $T$-action on \begin{align*}
Z_{c,\alpha} = \mu_{T,c}^{-1}(\alpha) = \mathcal{O}_c \cap \mu^{-1}(\alpha)
\end{align*}
is free and proper. Let $\zeta: Z_{c,\alpha}\hookrightarrow \mathcal{O}_c$ be the inclusion, and let
$\rho: Z_{c,\alpha} \rightarrow Z_{c,\alpha}/T$ be the principal $T$-bundle. Then the pullback of the KKS form to $Z_{c,\alpha}$ gives \begin{align*}
\iota_{X_H}\,\zeta^*\omega_{\mathrm{KKS},c} =\zeta^*\big(\iota_{X_H}\omega_{\mathrm{KKS},c}\big) =\zeta^*\big(d\langle \mu_T,H\rangle\big) =d\big(\langle \mu_T,H\rangle\circ \zeta\big)=0,
\end{align*} since $\langle\mu_T,H\rangle$ is constant on the level set $Z_{c,\alpha}$. Moreover, $\zeta^*\omega_{\mathrm{KKS},c}$ is $T$-invariant as $\omega_{\mathrm{KKS},c}$ is $G$-invariant. Hence, $\zeta^*\omega_{\mathrm{KKS},c}$ is a basic $2$-form for the principal $T$-bundle $\rho$. Since $T$ acts freely and properly on $Z_{c,\alpha}$, by Marsden-Weinstein reduction, this action produces a smooth symplectic manifold
\begin{align*}
\mathcal{O}_c//_\alpha T  = Z_{c,\alpha}/T
\end{align*}
with a symplectic structure $\omega_{\mathrm{red}}$ induced by
\begin{align*}
\zeta^*\omega_{\mathrm{KKS},c} = \rho^*\omega_{\mathrm{red}}.
\end{align*}

Let $J_\alpha:\mathcal{O}_c//_\alpha T \rightarrow  \mathfrak{g}//T$. To show that $\mathcal{O}_c//_\alpha T $ is a symplectic leaf, it is sufficient to show that $J_\alpha$ is a Poisson immersion. Consider the GIT quotient $\chi_T: \mathfrak{g} \rightarrow \mathfrak{g}//T$. It is a Poisson submersion. The pullback on the coordinate algebra induces the following injection $\chi_T^*:\mathcal{O}(\mathfrak{g}//T) \xrightarrow{ \ \cong \ }S(\mathfrak{g})^T \hookrightarrow \mathbb{C}[\mathfrak{g}^*] \cong S(\mathfrak{g}) $. The composition  $  Z_{c,\alpha} \xrightarrow{\ \zeta \ } \mathcal{O}_c \xrightarrow{ \ \chi_T\vert_{\mathcal{O}_c} = \chi_T \ }   \mathfrak{g}//T$ is $T$-invariant, and therefore it factors through a continuous map $J_\alpha: \mathcal{O}_c//_\alpha T\rightarrow \mathfrak{g}//T$. In this way, we form a commutative diagram 
\[\begin{tikzcd}
	{Z_{c,\alpha}} & {\mathcal{O}_c \subset \mathfrak{g}} \\
	{\mathcal{O}_c//_\alpha T = Z_{c,\alpha}/T} & {\mathfrak{g}//T}
	\arrow["\zeta", from=1-1, to=1-2]
	\arrow["\rho"', from=1-1, to=2-1]
	\arrow["{ \chi_T}", from=1-2, to=2-2]
	\arrow["{J_\alpha}"', from=2-1, to=2-2]
    \arrow[from=1-1, to=2-2, phantom, "\circlearrowright" {anchor=center, scale=1.5, rotate=90}]
\end{tikzcd}\]with $J_\alpha \circ \rho = \chi_T \circ \zeta$. 
We first show that $dJ_\alpha : T\left(\mathcal{O}_c//_\alpha T\right) \rightarrow T\left(\mathfrak{g}//T\right)$ is injective. 
By the definition of the orbit space, the fiber of $\rho$ through $u$ is exactly the orbit. That is, $\rho^{-1}([u]_T) = T\cdot u$.  Fix $u\in Z_{c,\alpha}$ and let $[u]_T:=\rho(u)\in \mathcal{O}_c//_\alpha T$. Since $\rho$ is a principal $T$-bundle, each fiber $\rho^{-1}([u]_T) = T\cdot u$ is an embedded submanifold, and   \begin{align*}
    \ker(d\rho)_u = T_u(T\cdot u).
\end{align*} Here $T \cdot u = \{t \cdot u:  t \in T\} = \rho^{-1}([u]_T) \subset Z_{c,\alpha}$ is the $T$-orbit through $u$, and $T_u(T\cdot u)$ denotes its tangent space. Consequently, by the first isomorphism theorem on the Poisson projection $d\rho$, there is a canonical identification \begin{align}
T_{[u]_T}(\mathcal{O}_c//_\alpha T) \cong  T_uZ_{c,\alpha}\big/ T_u(T\cdot u). \label{eq:tangentquo}
\end{align}

We now provide the differentiation of $J_\alpha$. Since $\zeta$ is an inclusion, for every $v\in T_uZ_{c,\alpha}$, $(d \zeta)_u: T_u Z_{c,\alpha} \to T_u \mathcal{O}_c$ is the inclusion of tangent spaces. That is, $(d\zeta)_u(v) = v$. Then the differential of the commutative relations above is as follows: \begin{align}
(dJ_\alpha)_{[u]_T}\big((d\rho)_u(v)\big) =(d\chi_T)_u(v).\label{eq:dJalpha}
\end{align} To conclude that $dJ_\alpha$ at $[u]_T$ is injective, it is sufficient to show that $\ker (d \chi_T)_u \subset \ker (d \rho)_u $. As $\chi_T$ is $T$-invariant, $(d\chi_T)_u$ vanishes on $T_u(T\cdot u)$, the right hand side of \eqref{eq:dJalpha} depends only on the class of $v$ modulo $T_u(T\cdot u)$. Thus, \eqref{eq:dJalpha} is compatible with \eqref{eq:tangentquo}. Suppose $(dJ_\alpha)_{[u]_T}(\varsigma)=0$ for some $\varsigma\in T_{[u]_T}(\mathcal{O}_c//_\alpha T)$. Let $V = (\mathfrak{g}//T)_\mathrm{reg}$. By Proposition \ref{prop:smoothregularg}, the restriction $\chi_T:\chi_T^{-1}(V)\to V$ is smooth. Hence, $(d\chi_T)_u$ is surjective for all $u\in\chi_T^{-1}(V)$. 
Choose $v\in T_uZ_{c,\alpha}$ with $\varsigma=(d\rho)_u(v)$. Since $\chi_T$ is $T$-invariant, the identity \eqref{eq:dJalpha} gives $(d\chi_T)_u(v)=0$. Hence, $v\in\ker(d\chi_T)_u$, and $$T_u(T\cdot u)\subseteq \ker(d\chi_T)_u.$$  Assume now that $J_\alpha([u]_T)\in V$. Then $u\in \chi_T^{-1}(V)$. In particular, $\chi_T$ is a submersion at $u$. Therefore, the fiber $ F:=\chi_T^{-1}(\chi_T(u))$ is a smooth embedded submanifold of $\mathfrak{g}$, and the constant rank theorem yields \begin{align}
T_uF=\ker(d\chi_T)_u . \label{eq:equalsfibers}
\end{align} Since $T$ is compact, invariant polynomials separate $T$-orbits. Hence, the fiber $F$ coincides with the $T$-orbit through $u$. That is, $ F=T\cdot u $. Hence, \eqref{eq:equalsfibers} immediately gives \begin{align}
\ker(d\chi_T)_u = T_u(T\cdot u). \label{eq:equaltangents}
\end{align} Finally, since $\rho: Z_{c,\alpha}\longrightarrow \mathcal{O}_c//_{\alpha}T$ is a principal $T$-bundle, we have $\ker(d\rho)_u = T_u(T\cdot u)$. Thus, if $v\in \ker(d\chi_T)_u$, then by \eqref{eq:equaltangents} we obtain $v\in T_u(T\cdot u)=\ker(d\rho)_u$, as claimed.
Therefore, $\varsigma=(d\rho)_u(v)=0$ implies that $(dJ_\alpha)_{[u]_T}$ is injective. Since $[u]_T$ was arbitrary, $J_\alpha$ is an immersion.

Define the pullback of $J_\alpha$ by \begin{align*}
    J_\alpha^*(f) \circ \rho = \left.(\chi_T^* f)\right\vert_{Z_{c,\alpha}}
\end{align*} for any $f \in \mathcal{O}\left(\mathfrak{g}//T\right)$. Let $\{\cdot,\cdot\}_{\sharp}$ be the Poisson structure on $\mathcal{O}\left(\mathfrak{g}//T\right)$. Using the fact that $\chi_T$ is Poisson, we compute \begin{align}
\rho^*\big(J_\alpha^*\{f,g\}_{\sharp}\big) &=(\chi_T\circ \zeta)^*\{f,g\}_{\sharp} =\zeta^*\big(\chi_T^*\{f,g\}_{\sharp}\big) =\zeta^*\big\{\chi_T^*f,\chi_T^*g\big\}_{\mathcal{O}_c}. \label{eq:rhoa}
\end{align} 
Comparing the \eqref{eq:rhoa}  with the relations $\rho^*\omega_{\mathrm{red}} = \zeta^* \omega_{\mathrm{KKS},c}$ yields
\begin{align*} 
\rho^*\big(J_\alpha^*\{f,g\}_{\sharp}\big)=\rho^*\big\{J_\alpha^*F,J_\alpha^*G\big\}_{\mathrm{red}}.
\end{align*} Here $\{\cdot,\cdot\}_{\mathrm{red}}$ is the Poisson bracket induced by the symplectic structure $\omega_{\mathrm{red}}$.
Since $\rho$ is a surjective submersion (on algebras), and $\rho^*$ is injective (on functions), the equality holds before applying $J_\alpha^*$ as follows:
\begin{align*}
J_\alpha^*\{f,g\}_{\sharp}=\{J_\alpha^*F,J_\alpha^*G\}_{\mathrm{red}}.
\end{align*}

Let $\pi_{\mathrm{red}}$ be the Poisson tensor on $\mathcal{O}_c//_\alpha T$, and let $\pi_\sharp$ be the Poisson tensor on $\mathfrak{g}//T$. Since $J_\alpha$ preserves the Poisson bracket (by the Poisson property), we have that the pushforward of the reduced Poisson tensor equals the Poisson tensor restricted to the image. That is, $${(J_\alpha)}_* \pi_{\mathrm{red}} = \left.\pi_{\sharp} \right\vert_{\mathrm{Im} \, J_\alpha}.$$ Hence, $\mathrm{Im}J_\alpha$ is a symplectic submanifold of $\mathfrak{g} //T$. This can be seen as follows: 
for any $u \in Z_{c,\alpha}$, we have \begin{align*}
    J(J_\alpha([u]_{\mathrm{red}})) = (c(u),u\vert_\mathfrak{t}) = (c,\alpha).
\end{align*} Therefore, $\mathrm{Im} \, J_\alpha \subset J^{-1}(c,\alpha)$. Conversely, if $[u]_T \in J^{-1}(c,\alpha)$, then $u \in \mathcal{O}_c$ and $u\vert_\mathfrak{t} = \alpha$, so $J_\alpha([u]_{\mathrm{red}}) \in \mathrm{Im} \, J_\alpha$. Thus, we have $\mathrm{Im} \, J_\alpha  = J^{-1}(c,\alpha)$, and each connected component of this fiber is a symplectic leaf in $\left(\mathfrak{g}//T\right)_{\mathrm{reg}}$.
\end{proof}
\begin{remark}
    If either $T$-action is not free or $\alpha$ is not regular, then $\mathcal{O}_c//_\alpha T$ is a stratified symplectic space in the sense of \cite{sjamaar1991stratified}.
\end{remark}
 
By Theorem \ref{thm:superalchain}, we have \begin{align*}
\dim \mathfrak{g}//T =\dim\mathfrak{g} - r.
\end{align*} Fixing the $2r$ central coordinates  $(C_1,\dots,C_r,\bar{\mu}_1,\dots,\bar{\mu}_r)$ yields the following:

\begin{corollary}
\label{cor:dimension}
Assume that $c$ and $\alpha$ are regular, and the $T$-action is free in $(\mathfrak{g}//T)_{\mathrm{reg}}$. Then $  \dim \mathcal{L}_{c,\alpha} = \dim\mathfrak{g} - 3r$. 
\end{corollary}
\begin{proof}
For a regular $c$, $\mathcal{O}_c$ has a dimension $\dim\mathfrak{g} - r$ as the stabiliser is a Cartan of dimension $r$. Moreover, if $\alpha$ is a regular value of $\mu_T$, and the $T$-action is free, then \begin{align*}
 \dim\bigl(\mathcal{O}_c//_{\alpha} T\bigr)=\dim\mathcal{O}_c - 2\dim T = (\dim\mathfrak{g} - r) - 2r = \dim\mathfrak{g} - 3r,
\end{align*}  which equals the dimension of the leaf $\mathcal{L}_{c,\alpha}$, is given by Theorem \ref{thm:leaves}.
\end{proof}
\begin{remark}
(i) In irregular strata (e.g., non-maximum rank of $dJ$), the foliation stratifies further, and the dimensions of the leaves may decrease.

(ii) 
By Definition \ref{def:superge}, together with Theorem \ref{thm:leaves}, we immediately conclude that \begin{align*}
    \left(\mathfrak{g}//T\right)_{\mathrm{reg}} = \bigsqcup_{(c,\alpha) \in J\left(\left(\mathfrak{g}//T\right)_{\mathrm{reg}} \right)} \mathcal{L}_{c,\alpha}.
\end{align*} For a fixed $(c,\alpha) \in \left(\mathfrak{g}//G\right) \times \mathfrak{t} $ with $c = (c_1,\ldots,c_r)$. Taking $H_1,\ldots,H_r$ to be the basis elements of $\mathfrak{t}$. Define the Poisson ideal \begin{align}
    \mathcal{I}_{c,\alpha} := \left\langle \rho^*(C_1) - c_1,\ldots,\rho^*(C_r) -c_r, \ \mu(H_1) -\alpha(H_1),\ldots, \mu(H_r) -\alpha(H_r)\right\rangle \label{eq:ideal}
\end{align} in $S(\mathfrak{g})^T$. It is clear that  $\mathcal{I}_{c,\alpha}$ is a Poisson ideal, as it is generated by fixing the Casimir elements. Then by Theorem \ref{thm:leaves}, we could define the following fiber algebra $\mathcal{P}_{c,\alpha} =  S(\mathfrak{g})^T \big/\mathcal{I}_{c,\alpha} $ such that \begin{align*}
    \mathrm{Spec} \, \mathcal{P}_{c,\alpha} =  \left(J^{-1}(c,\alpha)\right)_{\varpi} = \mathcal{L}_{c,\alpha} .
\end{align*}
\end{remark}

\section{Examples}
\label{sec:example}
 In this Section \ref{sec:example}, we will present an example to describe the superintegrable system chain presented above. Consider that \begin{align}
     G = SL_n(\mathbb{C}), \qquad \mathfrak{g}=\mathfrak{sl}_n(\mathbb{C}), \qquad T=\left\{\mathrm{diag}(t_1,\dots,t_n)\in SL_n(\mathbb{C})\right\}.
 \end{align} Let \begin{align*}
     \mathfrak{t} = \mathrm{Lie} (T) , \qquad H_i = E_{ii} - E_{i+1,i+1}, \quad 1 \leq i \leq n-1,
 \end{align*} and the linear coordinate functions on $\mathfrak{g}$ is denoted by $h_1,\ldots,h_{n-1}$, where $E_{ij}$ ($i \neq j$) is the elementary matrix. The symmetric algebra is then given by \begin{align}
     S(\mathfrak{g}) = \mathbb{C} \left[h_1,\dots,h_{n-1},e_{ij}\mid i\neq j\right].
 \end{align}
 
 We now consider the torus case, where $A = T$. This case is mainly considered, and its algebraic construction is provided in \cite{campoamor2023algebraic}. Here, we will provide an alternative proof. For pairwise distinct indices $i_1,\ldots,i_d$ with $1 \leq d \leq n$, define the cycle monomial \begin{align*}
     p_{i_1,\dots,i_d}:= e_{i_1 i_2}e_{i_2 i_3}\cdots e_{i_{d-1} i_d}e_{i_d i_1}.
 \end{align*}
\begin{proposition}
    Let $S(\mathfrak{g})^T$ be the $T$-invariant polynomial Poisson subalgebra of $S(\mathfrak{g})$. Then \begin{align*}
        S(\mathfrak{g})^T \cong \mathbb{C}[h_1,\ldots,h_{n-1},p_{i_1,\ldots,i_d} : 2 \leq d \leq n] \big/ \mathcal{I} ,
\end{align*} where $\mathcal{I}$ is generated by the following relation \begin{align*}
    \prod_{u=1}^{k-1} p_{i_u,i_{u+1}}\, p_{i_k,i_1}
 &= p_{i_1,i_2,\ldots,i_k}\, p_{i_1,i_k,i_{k-1},\ldots,i_2},
 \notag\\
 \prod_{1 \leq i < j \leq n} p_{i,j}
 &= p_{1,2,\ldots,n}\, p_{1,n-1,\ldots,2}\,
 \prod_{m=1}^{n-2}\ \prod_{s=m+2}^{n} p_{m,s},
 \notag\\
 \prod_{i_1 \neq i_2 \neq \cdots \neq i_k} p_{i_1,i_2,\ldots,i_k}
 &= \left(\prod_{r<s} p_{r,s}\right)^{\varphi(k)},
 \qquad
 \varphi(k)=\prod_{s=2}^{k-1}(n-s). 
\end{align*}
\end{proposition}
\begin{proof}
  Since $T$ is connected, Proposition \ref{prop:equ} implies that $S(\mathfrak{g})^T = S(\mathfrak{g})^\mathfrak{t}$. Since $\mathfrak{t}$ is the Cartan subalgebra of diagonal traceless matrices. Write  \begin{align*}
      H=\mathrm{diag}(\lambda_1,\dots,\lambda_n)\in t, \qquad \lambda_1 +\cdots + \lambda_n = 0,
  \end{align*} such that the scalars $\lambda_i$ are the diagonal entries of $H$. Recall that, for an elementary matrices $e_{ij}$ ($i \neq j$), we have $[H,e_{ij}] = (\lambda_i - \lambda_j) e_{ij}$. Hence, by the PBW theorem, the monomial $p = \prod_{ i \neq j} e_{ij}^{m_{ij}} \in S(\mathfrak{g})^T$, then using the Leibniz rule, we have \begin{align*}
      \{H,p\}&=\sum_{i\neq j} m_{ij}\,\{H,e_{ij}\}\,e_{ij}^{m_{ij}-1}\prod_{(a,b)\neq(i,j)} e_{ab}^{m_{ab}}\\
&=\left(\sum_{i\neq j} m_{ij}(\lambda_i-\lambda_j)\right)p = 0.
  \end{align*} It is convenient to regroup the coefficient of $M$ as
\begin{align*}
\sum_{i\neq j} m_{ij}(\lambda_i-\lambda_j)
=\sum_{k=1}^n \lambda_k\Bigl(\underbrace{\sum_{j\neq k} m_{kj}}_{\text{out}(k)}-\underbrace{\sum_{i\neq k} m_{ik}}_{\text{in}(k)}\Bigr).
\end{align*}
Since $H$ varies over all diagonal traceless matrices, the scalars $\lambda_1,\dots,\lambda_n$ are arbitrary, subject only to $\sum_k\lambda_k =0$. Therefore, $p\in S(\mathfrak{g})^T$ if and only if  \begin{equation}
    \sum_{j\neq k}m_{kj}=\sum_{i\neq k}m_{ik} \label{eq:balanceequation}
\end{equation}   for every $k\in\{1,\dots,n\}$.

Define the root lattices by $R = \bigoplus_{i\neq j}\mathbb{Z}\,\varepsilon_{ij}$. Define also a $\mathbb{Z}$-linear map \begin{align*}
    \delta: R \rightarrow \mathbb{Z}^n, \quad \delta(\varepsilon_{ij}) = e_i  - e_j.
\end{align*} It is clear that $\delta$ is a group homomorphism.  Then \begin{align*}
\delta\Bigl(\sum_{i\neq j} m_{ij} \varepsilon_{ij}\Bigr) = 0
\end{align*} if and only if \eqref{eq:balanceequation} holds. For pairwise distinct indices $i_1,\dots,i_d$ ($d \geq 2$) define the cycle element by \begin{align*}
c_{i_1,\dots,i_d} := \varepsilon_{i_1 i_2} + \varepsilon_{i_2 i_3}+ \cdots+\varepsilon_{i_d i_1}\in R.
\end{align*} A direct telescoping computation gives $\delta(c_{i_1,\dots,i_d}) = 0$, so $c_{i_1,\dots,i_d}\in\ker(\delta)$.

We claim that any vector $$m := \sum_{i\neq j} m_{ij}\varepsilon_{ij}\in \ker(\delta)$$ with $m_{ij}\geq 0$ is a sum of such cycle elements by induction on $|m| := \sum_{i\neq j} m_{ij}$. If $|m| = 0$, there is nothing to prove. Otherwise, choose $i_1\neq i_2$ with $m_{i_1 i_2}>0$. Suppose $i_r$ has been chosen. The $i_r$-th balance equation implies that the total outgoing multiplicity from $i_r$ equals the total incoming multiplicity. Hence, in particular, there exists some $i_{r+1}\neq i_r$ with $m_{i_r i_{r+1}}>0$. Continuing in this way produces a sequence $i_1,i_2,i_3,\dots$. Since $\{1,\dots,n\}$ is finite, some index repeats: $i_s=i_t$ with $s<t$. If the repeated segment contains further repetitions, delete subloops to obtain a subsequence with pairwise distinct indices. In either case, we obtain a cycle $(j_1, \dots,j_d)$ with pairwise distinct entries such that each edge $(j_r \to j_{r+1})$ (with $j_{d+1} = j_1$) occurs with positive multiplicity in $m$. Therefore, $m - c_{j_1,\dots,j_d}$ still has nonnegative coefficients and still lies in $\ker (\delta)$ (since $\delta (c_{j_1,\dots, j_d})=0$). By induction, $m-c_{j_1,\dots,j_d}$ is a sum of cycle elements, hence so is $m$.

Translating back, subtracting $c_{j_1,\dots,j_d}$ corresponds to factoring $M$ by the cycle monomial $p_{j_1,\dots,j_d}$. Thus, any $T$-invariant monomial $M$ is a product of cycle monomials $p_{i_1,\dots,i_d}$ with pairwise distinct indices.

Since the action of $T$ on $t$ is trivial, every polynomial in $h_1,\dots,h_{n-1}$ is $T$-invariant. It follows that
$S(\mathfrak{g})^T$ is generated by $h_1,\dots,h_{n-1}$ together with the cycle monomials.
\end{proof}

   \begin{remark}
       We interpret the generators $p_{i_1,\dots,i_d}=e_{i_1 i_2}e_{i_2 i_3}\cdots e_{i_{d-1} i_d}e_{i_d i_1}$ as \emph{loop} (or \emph{cycle}) monomials: the ordered list $(i_1\,i_2\,\dots\,i_d)$ encodes the directed cycle $i_1\to i_2 \to \cdots \to i_d\to i_1$.

A \emph{subloop} of such a loop is any smaller directed cycle obtained from a proper subset of the vertices, i.e., a cycle $(j_1\,\dots\,j_m)$ with $2 \leq m < d$ and $\{j_1,\dots,j_m\}\subset\{i_1, \dots, i_d\}$ for which the corresponding monomial $p_{j_1,\dots,j_m}$ is formed by the edges among these vertices. Equivalently, the presence of a subloop means that (after reordering commuting factors) the associated monomial is \emph{decomposable}, i.e., it factorizes as a product of loop monomials of strictly smaller lengths.

The $T$-invariant subalgebra $S(\mathfrak{g})^T$ is then generated by the Cartan variables $h_1,\dots,h_{n-1}$ together with these loop monomials, modulo the algebraic relations among them, collected in the ideal $\mathcal{I}$.
   \end{remark}

We now have a direct application of Theorem \ref{thm:superalchain}.  Let \begin{align*}
c_k(X):=\mathrm{tr}(X^k), \qquad 2\leq k\leq n.
\end{align*} Then  $S(\mathfrak{g})^G = \mathbb{C}[c_2,\dots,c_n]$.  Moreover, \begin{align*}
\mathrm{trdeg}\, S(\mathfrak{g})^G=n-1, \qquad \mathrm{trdeg}\, S(\mathfrak{g})^T=(n^2-1)-(n-1)=n(n-1).
\end{align*} Hence, the chain  $\mathfrak{g} \xrightarrow{\ \chi_T\ } \mathfrak{g}//T \xrightarrow{\ \rho\ } \mathfrak{g}//G$ is superintegrable.

\section{Conclusion}
\label{sec:conclusion}

In this work, we developed a Poisson-algebraic framework for constructing $A$-invariant functions in $S(\mathfrak{g})$ and for describing the associated Poisson projection chain.  The central point is that the relevant families of first integrals are organized by morphisms of affine Poisson varieties, so that superintegrability can be checked by intrinsic dimension identities along the chain rather than by model dependent computations.  Concretely, we introduced the notion of a superintegrable Poisson chain, constructed Poisson subalgebras arising from reductive data, and made explicit the corresponding Poisson projections between quotients.  This yields a uniform geometric explanation for why reductive chains produce superintegrable systems and recovers the previously known special cases within a single conceptual picture.

On the geometric side, the projection-chain viewpoint clarifies how invariant theory controls the induced foliations: the base algebra governs the Casimir directions, while the intermediate quotients encode the additional commuting or noncommuting integrals that generate the superintegrable distribution.  In particular, the resulting description of symplectic leaves and their fibers provides a convenient starting point for comparing with classical constructions such as Thimm-type subalgebra chains and Mishchenko-Fomenko noncommutative integrability, and for extracting explicit generators (as illustrated in the $\mathfrak{sl}(n,\mathbb{C})$ examples).

Several directions remain open.  First, beyond the cases $A=T$ and $A$ Abelian, it would be interesting to identify reductive subgroups $A\subset G$ (or more general $A$-stable Poisson subalgebras) for which the superintegrability dimension conditions fail and to characterize the failure in terms of the algebraic relations, the singular locus of the quotients, and the degeneration of generic fibers.  Second, one can ask for global refinements: a description of the connected components of generic fibers, the completeness of the Hamiltonian flows produced by the chain, and the behaviour under restriction to real forms and compact forms.

Third, a natural continuation of the present work is quantization.  One expects a parallel construction at the level of $U(\mathfrak{g})$ (or of quantized coordinate rings) producing large commutative and/or Gelfand-Tsetlin type subalgebras, together with a precise comparison between their semiclassical limits and the Poisson projection chains constructed here.  Finally, it would be interesting to extend the projection-chain formalism to other Lie-theoretic Poisson manifolds (e.g., coadjoint orbits, multiplicative Poisson-Lie settings, or moduli-type spaces) where similar reductive data should produce new superintegrable families.  We plan to address these problems in future work.

\section*{Acknowledgement}

Kai Jiang was supported by the Fundamental Research Funds for the Central Universities. Guorui Ma was supported by the Research Funding at SIMIS.   Ian Marquette was supported by the Australian Research Council Future Fellowship FT180100099. Yao-Zhong Zhang was supported by the Australian Research Council Discovery Project DP190101529.    The authors would like to thank Nicolai Reshetikhin, and Zhuo Chen for useful discussions on the geometric construction of the superintegrabilities, and the construction of the homogeneous spaces with He, Lyu.

   \bibliographystyle{unsrt}
\bibliography{bibliography.bib}
\end{document}